\documentclass[twocolumn,aps,american,pra,floatfix,nofootinbib,superscriptaddress,longbibliography,10pt]{revtex4-1}

\usepackage{amsfonts}
\usepackage{amsmath}
\usepackage{amssymb}
\usepackage{amsthm}
\usepackage{bbm}
\usepackage{bm}
\usepackage{mathrsfs}
\usepackage{MnSymbol}
\usepackage{accents}
\usepackage{braket}
\usepackage{multirow}
\usepackage[table]{xcolor}
\usepackage{rotating}  % in preamble
\usepackage{float}
\usepackage{soul}

\usepackage{epsfig}
\usepackage{graphics}
\usepackage{graphicx}
\usepackage{subfigure}
\usepackage{tikz}
\usetikzlibrary{arrows.meta}
\usetikzlibrary{quotes,fit}
\usetikzlibrary{datavisualization.formats.functions}
\usepackage[matrix,frame,arrow]{xy}
\usepackage[compat=0.6]{yquant}
\useyquantlanguage{groups}
\usepackage{times}
\usepackage{algorithm}
\usepackage{algpseudocode}
\algrenewcommand\algorithmicrequire{{Data:}}
\algrenewcommand\algorithmicensure{Initialize:}
\usepackage{color}
\usepackage{xcolor}

\usepackage[pdfstartview=FitH]{hyperref}
\usepackage{hypernat}
\hypersetup{
    colorlinks=true,   % false: boxed links; true: colored links
    linkcolor=cyan,    % color of internal links
    citecolor=magenta, % color of links to bibliography
    filecolor=magenta, % color of file links
    urlcolor=cyan,     % color of external links
    runcolor=cyan
}
\usepackage[capitalise]{cleveref} % Should be loaded after hyperref

\newtheorem{mytheorem}{Theorem}
\newtheorem{mylemma}{Lemma}

\newcommand{\mc}{\mathcal}
\newcommand{\wSimon}[2]{\ensuremath{\mathsf{w}_{{#1}}\textnormal{Simon-}{#2}}}
\newcommand{\NTS}{\mathrm{NTS}}
\newcommand{\NTSClb}{\NTS_C^{\mathrm{lb}}}
\newcommand{\NTSIQ}{\NTS_{\textnormal{IQ}}}
\newcommand{\EEB}{E_{\textnormal{EB}}}
\newcommand{\beq}{\begin{equation}}
\newcommand{\eeq}{\end{equation}}

\newcommand{\mcO}{\mathcal{O}} % $\mcO$ is a box labeling a place to insert the oracle.
\newcommand{\mcOf}{\mathcal{O}_f} % $\mcOf = \mcO^f$ is the unitary implementing the actual oracle for function $f$.

\newcommand{\expv}[1]{\langle #1\rangle} % expectation value
\DeclareMathOperator{\HW}{HW} % $\HW(b)$ is the Hamming weight of a bitstring $b$.

\def\Pr{\mathrm{Pr}}

\graphicspath{ {images/} }

\makeatletter
\AddToHook{cmd/appendix/before}{\def\cref@section@alias{appendix}\def\cref@subsection@alias{appendix}}
\makeatother

\begin{document}

\title{
Demonstration of Exponential Quantum Speedup with Constant-Depth Compiled Circuits for Simon's Problem}

\author{Phattharaporn Singkanipa}
\affiliation{Department of Physics, Center for Quantum Information Science \& Technology,
University of Southern California, Los Angeles, CA 90089, USA}

\author{Victor Kasatkin}
\affiliation{Viterbi School of Engineering, Center for Quantum Information Science \& Technology,
University of Southern California, Los Angeles, CA 90089, USA}

\author{Daniel A. Lidar}
\affiliation{Departments of Electrical \& Computer Engineering, Chemistry, Physics \& Astronomy, and Center for Quantum Information Science \& Technology,
University of Southern California, Los Angeles, CA 90089, USA}
\affiliation{Quantum Elements, Inc., Thousand Oaks, CA}

\date{\today}

\begin{abstract}
We demonstrate exponential algorithmic quantum speedup for a restricted-Hamming-weight
version of Simon's problem, in which the hidden string $b$ is promised to satisfy
$\HW(b)\le w$ for a Hamming-weight cutoff $w$,
on present-day superconducting quantum processors.
We introduce a hardware-aware compilation strategy that reduces
the quantum part of each Simon query circuit to constant depth.
The resulting compiled circuits have $O(1)$ depth,
require only linear nearest-neighbor connectivity,
map directly onto common device layouts,
and avoid additional routing and SWAP overhead.
Implemented on IBM's $156$-qubit Boston and $120$-qubit Miami processors,
these circuits achieve sufficient fidelity
to exhibit algorithmic quantum speedup without error suppression.
Using the number-of-queries-to-solution (NTS) metric,
we observe exponential speedup over the classical lower-bound benchmark
for all restricted-Hamming-weight cutoffs $w\ge 4$ on Boston
and across low-to-intermediate Hamming-weight cutoffs on Miami;
at higher Hamming-weight cutoffs on Miami, we still observe polynomial speedup.
The same construction also enables unrestricted instances of Simon's problem,
corresponding to $w=n$ for problem size $n$,
over the finite problem-size ranges for which our NTS computation is feasible;
in this regime, the observed scaling advantage is not limited to the restricted-Hamming-weight setting.
These results show that careful hardware-aware 
compilation can make quantum speedup experimentally accessible
for a canonical hidden-subgroup problem in the NISQ regime.
\end{abstract}

\maketitle
\section{Introduction}
Demonstrating an algorithmic quantum speedup,
one in which a quantum algorithm solves a computational problem
more efficiently than any classical algorithm as the problem size grows,
is a central goal of quantum computing.
While theoretical quantum advantages have been established for decades
\cite{Deutsch:92,bernsteinQuantumComplexityTheory1997,Simon:94,Grover:97a,Shor:97,childs2003exponential,Van-Dam:2006aa,Harrow:2009aa,montanaroQuantumAlgorithmsOverview2016,Bravyi:2017aa,Bravyi:2020aa,Bharti:2022aa,Daley:2022vu},
realizing them on physical hardware is challenging due to the noise
and limited scale of current noisy intermediate-scale quantum (NISQ) \cite{Preskill2018} devices.
This goal is distinct from quantum supremacy experiments
\cite{aaronson2016,Arute:2019aa,wu2021strong,Zhong:2020aa,Zhong:2021wv,morvan2023phase},
which demonstrate classically intractable sampling tasks without addressing computational utility,
and from device benchmarking
\cite{figgattComplete3QubitGrover2017,wrightBenchmarking11qubitQuantum2019,royProgrammableSuperconductingProcessor2020,Pelofske2022,Lubinski2021,doi:10.1021/acs.chemrev.4c00870},
which characterizes hardware performance without resolving the question of algorithmic scaling
\cite{Barak:spoofing,Zlokapa:2023aa,Aharonov:22}.

Experimental progress toward this goal has been made across a range of hardware platforms
\cite{Albash:2017aa,King:2019aa,Saggio:2021vh,Centrone:2021tq,Maslov:2021aa,Xia:2021ux,Huang:2021,Ebadi:22,zhouExperimentalQuantumAdvantage2022,King:22,Kim:2023aa},
where better-than-classical success probabilities have been reported
for specific problem instances.
However, many of these demonstrations compare against a restricted set
of classical algorithms or rely on complexity-theoretic assumptions
rather than direct empirical evidence of a scaling advantage.
A rigorous demonstration of algorithmic quantum speedup
requires a performance comparison against the best known classical algorithms
over a broad and increasing range of problem sizes,
which remains a stringent benchmark that few experiments have met.

Conjecture-free quantum scaling advantages have recently been demonstrated
on IBM Quantum superconducting processors using hardware-aware error suppression techniques.
Ref.~\cite{pokharel2022demonstration} reported a scaling advantage for the
Bernstein-Vazirani algorithm using dynamical decoupling (DD)
\cite{Viola:98,Viola:99,Zanardi:1999fk,Vitali:99,Duan:98e},
while Ref.~\cite{PhysRevX.15.021082} extended this approach
to restricted-Hamming-weight instances of Simon's problem,
which belongs to the class of Abelian hidden subgroup problems \cite{Jozsa:2001aa}.
The compiled Simon's problem circuit in Ref.~\cite{PhysRevX.15.021082}
had a star-type CNOT structure with depth $O(\HW(b))$ for a hidden string $b$.
Embedding this structure onto the heavy-hex connectivity architecture of IBM Quantum processors
introduced substantial SWAP overhead,
making error suppression indispensable for maintaining circuit fidelity.
Specifically, DD was combined with measurement error mitigation
(MEM) \cite{PhysRevLett.119.180509} to demonstrate algorithmic speedup
in the resulting deep circuits.

In this work, we show that,
under the same compiler rules as in Ref.~\cite{PhysRevX.15.021082},
Simon query circuits can be compiled so that their quantum part has constant depth, without changing the underlying oracle problem.
Our construction requires only linear connectivity and yields compiled query circuits with $O(1)$ quantum depth, independent of $\HW(b)$.
These circuits use only two entangling layers, map natively onto common two-dimensional superconducting-qubit layouts,
and avoid additional routing and SWAP overhead.
Our results show that careful hardware-aware compilation can significantly
enhance quantum-classical scaling separations on NISQ devices,
even without error suppression,
and underscore the importance of co-designing algorithms with hardware constraints in mind.

In the original formulation of Simon's problem,
one is given a function $f \colon \{0,1\}^n \to \{0,1\}^n$
promised to be either 1-to-1 or 2-to-1,
such that $f_b(x) = f_b(y)$ if and only if $x = y$ or $x = y \oplus b$,
where $b \in \{0,1\}^n$ is a hidden bitstring.
The task is to determine which case holds and,
in the 2-to-1 case, to identify $b$.
Here we focus on the 2-to-1 case and restrict $b$ to have Hamming weight at most $w \le n$,
yielding the restricted-HW version of Simon's problem \wSimon{w}{n}.
The original Simon's problem is recovered by setting $w = n$.

We analyze performance using an oracle-query metric, the number-of-queries-to-solution (NTS)~\cite{PhysRevX.15.021082},
and show that this setting enables a clear quantum scaling advantage even on contemporary superconducting qubit architectures.
The same compiled-circuit construction also allows the original Simon's problem to be implemented up to a certain problem size,
so the observed scaling advantage is not limited to the restricted-Hamming-weight setting in those cases.

The speedup is quantified in a guessing game setting introduced in \cite{PhysRevX.15.021082},
whose rules are briefly summarized in \cref{sec:speedup-setup}.
In particular, we restate the game settings, the speedup metric NTS,
and the classical and quantum algorithms for solving Simon's problem,
along with the methodology for comparing the performance of each player.
In \cref{sec:o1-oracle}, we present the main contribution of this work:
a compilation scheme that yields constant-depth Simon query circuits,
together with its experimental implementation.
The speedup results are presented in \cref{sec:results},
followed by the conclusion in \cref{sec:conclusion}.
Additional technical details and supplemental results are provided in the appendices.

\section{Quantum speedup in a guessing game}\label{sec:speedup-setup}

\subsection{Rules of the game}
We adopt the single-player guessing game introduced in Ref.~\cite{PhysRevX.15.021082}. 
At the start of each round, a hidden string $b$ is drawn uniformly at random from the allowed set $S$ of nonzero hidden strings for the problem under consideration. Conditioned on this $b$, a function $f\colon \{0,1\}^n \to \{0,1\}^n$ is drawn uniformly at random from the set of 2-to-1 functions satisfying $f(x)=f(y)$ if and only if $x=y$ or $x=y\oplus b$. The function is not revealed to the player, who instead has oracle access to evaluate $f(x)$ for any input $x \in \{0,1\}^n$.
The player may perform any number of oracle queries, along with arbitrary classical computation, before submitting a guess for the hidden bitstring $b$. The correctness of the guess is then verified, after which a new function is drawn and the game proceeds to the next round. The objective of the player is to maximize the number of correct guesses while minimizing the total number of oracle queries.

The performance of each player is quantified by the average score per oracle query, $\NTS^{-1}$, with a higher value indicating better performance, where
\begin{equation}
  \NTS = \frac{\expv{Q}}{\expv{P}} .
  \label{eq:NTS.def}
\end{equation}
Here, NTS denotes the \emph{number-of-oracle queries-to-solution}, $\expv{\bullet}$ denotes the expectation value averaged over many rounds of the game, $Q$ is the number of oracle queries per round, and $P$ is the score obtained in a round. To discourage uninformed guessing without consulting the oracle, incorrect guesses are penalized by setting $P = 1$ for a correct guess and $P = -p_r / (1 - p_r)$ for an incorrect guess, where $S$ denotes the allowed set of hidden strings and $p_r = 1/|S|$ is the probability that a uniformly random guess is correct. In the restricted-HW setting considered below, $|S| = N_w$.

Three players compete in this game: (i) a classical player, whose oracle access consists of submitting a query $x$ and receiving $f(x)$ in return; (ii) a noiseless quantum player, whose oracle is the unitary $\mcOf$ acting as $\mcOf \ket{x}\ket{a} = \ket{x}\ket{a \oplus f(x)}$;
and (iii) a NISQ player, who has indirect access to a noisy quantum processing unit (QPU) on which circuits containing $\mcOf$ are executed after compilation.

\subsection{Classical algorithm}
To generalize to \wSimon{w}{n}, we introduce $N_w$ as the cardinality of the set $S$ of possible values of $b \neq 0^n$, restricted to those satisfying $\HW(b) \le w$, with $w\le n$:
\begin{equation}
N_w \equiv \sum_{j=1}^{w}\binom{n}{j},
\label{eq:N_w}
\end{equation}
which simplifies to $N_n\equiv 2^n-1$ in the original Simon's problem.

\begin{mytheorem}[Lower bound on $\NTS_C$ \cite{PhysRevX.15.021082}]
  \label{thm:cl-lower-bound}
For the trivial case $N_w=1$, one has $\NTS_C=0$. For $N_w\ge 2$, the lower bound on the classical NTS is
\begin{equation}
\NTS_C \ge k_{\min}-\frac{k_{\min}(k_{\min}-1)(k_{\min}-2)}{6N_w},
    \label{eq:cl-expect2}
\end{equation}
where
\begin{equation}
     k_{\min}(N_w)\equiv \left\lceil \sqrt{2N_w-\frac{7}{4}}+\frac{1}{2}\right\rceil
     \label{eq:cl-worst}
\end{equation}
is a lower bound on the worst-case number of queries needed by a classical player to solve \wSimon{w}{n}.
\end{mytheorem}

Following this theorem, we define
the classical lower-bound quantity used in our scaling comparison as
\begin{equation}
  \NTSClb \equiv
  \begin{cases}
  0, & N_w=1,\\[2mm]
  k_{\min}-\dfrac{k_{\min}(k_{\min}-1)(k_{\min}-2)}{6N_w}, & N_w\ge 2.
  \end{cases}
  \label{eq:NTSClb-def}
\end{equation}

In this work we compare the scaling of the $\NTS_Q$ with
$\NTSClb$ instead of $\NTS_C$ because
(1) we do not know the exact optimal value of $\NTS_C$
(for the best possible classical player), and
(2) if the quantum algorithm defeats this lower bound,
the implication is an unequivocal quantum speedup.

For fixed $w$ (independent of $n$), the dominant term in $N_w=\sum_{j=1}^{w}\binom{n}{j}$ is $\binom{n}{w}$, so $N_w=\Theta(n^w)$ as $n\to\infty$.
Consequently, $k_{\min}=\Theta(n^{w/2})$, and the lower bound in \cref{eq:cl-expect2} scales as $n^{w/2}$.
Without restricting $\HW(b)$, the classical query complexity is $\Theta(2^{n/2})$.

\subsection{Noiseless quantum algorithm}
We now describe how Simon's problem can be solved on a noiseless quantum computer.
\Cref{fig:simon-circ} shows the quantum circuit implementing this algorithm, which employs $2n$ qubits for \wSimon{w}{n}:
the first $n$ are data qubits, labeled $d_j$, and the remaining $n$ are ancilla qubits, labeled $a_j$, for $j = 0, 1, \dots, n-1$.
The box $\mc{O}_f$ denotes the oracle implementing $f(x)$ on input $x$, whose internal circuit is not accessible to the players,
where the function $f$ is associated with some hidden bitstring $b$.
Each execution of the circuit produces a uniformly random bitstring $z$ satisfying $b \cdot z = 0$.
The standard proof that $O(n)$ executions of this circuit suffice, on average, to solve Simon-$n$ on an ideal quantum computer can be found, e.g., in Ref.~\cite[App.~E]{PhysRevX.15.021082}.

\begin{figure}[h!]
    \centering
    \begin{tikzpicture}
        \begin{yquant}
            qubit {$\ket{0}_{d_0}$} x1;
            nobit d0;
            qubit {$\ket{0}_{d_{n-1}}$} xn;
            qubit {$\ket{0}_{a_0}$} a1;
            nobit d1;
            qubit {$\ket{0}_{a_{n-1}}$} an;
            h x1;
            h xn;
            text {$\vdots$} d0;
            text {$\vdots$} d1;
            box {$\mc{O}_f$} (x1, d0, xn, d1, a1, an);
            h x1;
            h xn;
            measure a1;
            measure an;

            align x1, xn, a1, an;
            measure x1;
            measure xn;
            text {$\vdots$} d0;
            text {$\vdots$} d1;
        \end{yquant}
    \end{tikzpicture}
    \caption{Quantum circuit for solving Simon's problem with a length-$n$ hidden bitstring $b$. The structure of the oracle $\mc{O}_f$ is not visible to the player who wishes to guess $b$. The top $n$ measurement results form the bitstring $z$, and the bottom $n$ measurement results are discarded in the algorithm.}
    \label{fig:simon-circ}
\end{figure}
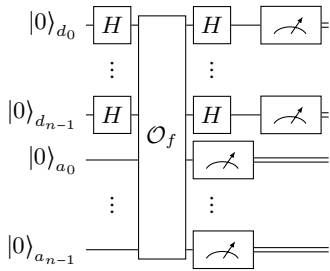

We require the $\NTS$ for the noiseless quantum \wSimon{w}{n} problem for an arbitrary $w$. More specifically, we are interested in the $\NTS$ of an optimal, ideal quantum player who runs the circuit in \cref{fig:simon-circ} until $b \in S$ is uniquely determined. We denote this quantity by $\NTSIQ$, short for ``ideal quantum.''

The exact value of $\NTSIQ$ is known in the limiting cases $w=1$ and $w=n$ \cite{PhysRevX.15.021082}. For intermediate $w$, we use the interpolation
\begin{equation}
  \NTSIQ(t) = \log_2(N_w) + \left(\frac12 + \frac{\gamma}{\ln(2)}\right) (1-t) + (\EEB - 1) t ,
  \label{eq:NTSIQ-interp}
\end{equation}
from \cite{PhysRevX.15.021082} as an accurate estimate of $\NTSIQ$,
where $t \equiv N_w / (2^n - 1)$ is the density of admissible $b$ values among all nonzero bitstrings,
$\gamma = 0.57721\dots$ is the Euler-Mascheroni constant,
and $\EEB = 1.60669\dots$ is the Erd\H{o}s--Borwein constant \cite{Erdos-Borwein}.

\subsection{NISQ algorithm}
Our NISQ player uses the circuit as on \cref{fig:simon-circ} (without knowing the contents of $\mc{O}_f$).
However, this is not the circuit executed on an actual device:
under the same rules for the compiler as in \cite{PhysRevX.15.021082},
the compiler can replace the deep oracle $\mc{O}_f$ of \cref{fig:simon-circ} by a simpler representative oracle $\mc{O}_{f_b}$. Using the freedom to relabel bit positions, we choose the canonical representative $b=0^{n-i}1^i$ for each Hamming-weight class $i=\HW(b)$. With this convention, the compiled representative circuit depends only on $n$ and $i$.
We restate the compilation rules in \cref{app:rules:oracle}.
The compilation and the representative circuits $\mc{O}_{f_b}$ are explained later in \cref{sec:o1-oracle} and \cref{app:simon-oracle}.
This simpler circuit is then executed on a noisy QPU.
This is the experimental scenario considered in this work, and the corresponding results are presented in \cref{sec:results}.
To quantify algorithmic performance via NTS, we define a practical $\NTS_Q$ that can be computed directly from experimental data as
\begin{equation}
\NTS_Q(n; w) = \frac{\langle Q \rangle}{\langle P \rangle} =\frac{N_w-1}{N_w} \frac{\sum_{i=1}^{w} h_i Q_i}{\sum_{i=1}^{w} h_i p_i - 1} ,
\label{eq:NTS-wn}
\end{equation}
where $h_i = \binom{n}{i}$ is the number of length-$n$ bitstrings with Hamming weight $i$, $Q_i$ is the number of oracle queries used for
a representative bitstring with Hamming weight $i$,
and $p_i$ is the corresponding success probability.
This expression assumes that, after compilation, the circuit performance depends only on $\HW(b)=i$,
so all weight-$i$ bitstrings are statistically equivalent.
For the NISQ player presented here, this assumption is enforced by compiling one canonical representative for each Hamming-weight class and using the compiler's qubit-relabeling freedom to treat all strings of the same Hamming weight as equivalent.

\subsection{Quantum speedup quantification}
We define an \emph{algorithmic quantum speedup}
for the $\wSimon{w}{n}$ problem as better scaling with $N_w$
of the function $\NTS_Q$ [\cref{eq:NTS-wn}]
than of the function $\NTS_C$ for the best possible classical strategy:
$\NTS_Q = o(\min_{\pi} \NTS_C(\pi))$, where the minimum is taken
over all classical strategies $\pi$.
As follows from \cref{thm:cl-lower-bound},
the sufficient condition for such speedup is scaling with $N_w$ better than $\NTSClb$
defined in \cref{eq:NTSClb-def}.
Since we do not know how exactly $\min_{\pi} \NTS_C(\pi)$ scales with $N_w$,
we use that sufficient condition in our analysis.
We determine the scaling by fitting two two-parameter models.
Both models satisfy the constraint $\NTS(N_w=1)=0$,
since when $N_w=1$ the only possible string is $b=1$, so no oracle calls are needed.

First, consider a polylogarithmic (polylog) model in $N_w$:
\begin{equation}
    \NTS_{\text{polylog}}=a(\log_2{N_w})^{\alpha} .
    \label{eq:polylog-NTS}
\end{equation}
Second, consider a polynomial (poly) model in $N_w$:
\begin{equation}
    \NTS_{\text{poly}}= b (N_w^\beta-1) .
    \label{eq:poly-NTS}
\end{equation}

Since the classical lower bound scales as $N_w^{1/2}$ and therefore follows the poly model,
an \emph{exponential} algorithmic quantum speedup is observed
when $\NTS_Q$ follows the polylog model.
If, instead, $\NTS_Q$ follows the poly model with fitted exponent $\beta_Q<\beta_C=1/2$ over the fitted range,
a \emph{polynomial} speedup in the $N_w$-based NTS metric is observed.
We refer to $\alpha$ and $\beta$ as the scaling exponents.
The fitting parameters $a$ and $b$ do not matter for scaling purposes.
We explain how we determine which model is the better fit in \cref{sec:results}.

\section{Constant-depth compilation of Simon query circuits}
\label{sec:o1-oracle}

We adopt the notation introduced in Ref.~\cite{PhysRevX.15.021082}
to represent qubits and CNOT gates,
where each node is labeled by its qubit name and a directed edge, represented using an arrow,
from qubit $A$ to qubit $B$ denotes a CNOT gate controlled on $A$ and targeting $B$.
A detailed construction of the 
circuit family underlying this constant-depth compilation is presented in
\cref{app:simon-oracle}, along with a proof 
that it implements a valid Simon oracle.

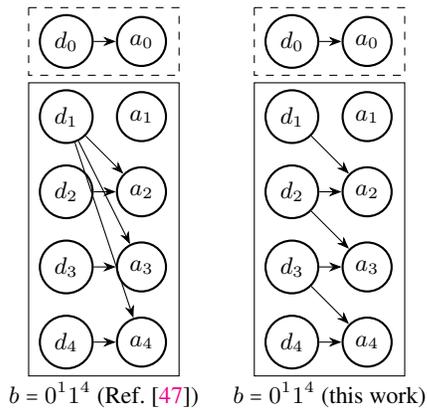
\begin{figure}[h!]
    \centering
    \begin{tikzpicture}
        \begin{scope}[every node/.style={circle,thick,draw}]
            \node (4x0) at (0,1.7) {$d_0$};
            \node (4x1) at (0,0.7) {$d_1$};
            \node (4x2) at (0,-0.3) {$d_2$};
            \node (4x3) at (0,-1.3) {$d_3$};
            \node (4x4) at (0,-2.3) {$d_4$};
            \node (4a0) at (1,1.7) {$a_0$};
            \node (4a1) at (1,0.7) {$a_1$} ;
            \node (4a2) at (1,-0.3) {$a_2$};
            \node (4a3) at (1,-1.3) {$a_3$};
            \node (4a4) at (1,-2.3) {$a_4$};

            \node (3x0) at (3,1.7) {$d_0$};
            \node (3x1) at (3,0.7) {$d_1$};
            \node (3x2) at (3,-0.3) {$d_2$};
            \node (3x3) at (3,-1.3) {$d_3$};
            \node (3x4) at (3,-2.3) {$d_4$};
            \node (3a0) at (4,1.7) {$a_0$};
            \node (3a1) at (4,0.7) {$a_1$} ;
            \node (3a2) at (4,-0.3) {$a_2$};
            \node (3a3) at (4,-1.3) {$a_3$};
            \node (3a4) at (4,-2.3) {$a_4$};
        \end{scope}
        \node (eqL) at (0.5,-3) {$b=0^11^4$ (Ref.~\cite{PhysRevX.15.021082})} ;
        \node (eqR) at (3.5,-3) {$b=0^11^4$ (this work)} ;
        \begin{scope}[>={Stealth[black]}]
            \draw[->] (4x0) -- (4a0);
            \draw[->] (4x1) -- (4a2);
            \draw[->] (4x1) -- (4a3);
            \draw[->] (4x1) -- (4a4);
            \draw[->] (4x2) -- (4a2);
            \draw[->] (4x3) -- (4a3);
            \draw[->] (4x4) -- (4a4);

            \draw[->] (3x0) -- (3a0);
            \draw[->] (3x1) -- (3a2);
            \draw[->] (3x2) -- (3a3);
            \draw[->] (3x3) -- (3a4);
            \draw[->] (3x2) -- (3a2);
            \draw[->] (3x3) -- (3a3);
            \draw[->] (3x4) -- (3a4);
        \end{scope}

        \draw[dashed] (-0.5,2.15) rectangle (1.5,1.25);
        \draw[-] (-0.5,1.15) rectangle (1.5,-2.75);

        \draw[dashed] (2.5,2.15) rectangle (4.5,1.25);
        \draw[-] (2.5,1.15) rectangle (4.5,-2.75);
    \end{tikzpicture}
    \caption{Comparison of two oracle constructions for Simon-$5$ ($b = 01111$) in the graph representation. 
CNOTs can be placed in the same entangling layer when they do not share any qubits.
    For the $\HW(b)=4$ instance shown here, the construction presented in this work 
    can be scheduled in two entangling layers,
    whereas the construction in Ref.~\cite{PhysRevX.15.021082} requires three under the same non-overlapping-qubit scheduling rule.
    }
    \label{fig:simon-compare}
\end{figure}

Our goal in this section is to illustrate the key improvement of our construction relative to that of Ref.~\cite{PhysRevX.15.021082}. This is shown in \cref{fig:simon-compare}, where both circuits implement a valid \wSimon{w}{n} oracle for $w = 4$ and $n = 5$. In both constructions, arrows denote CNOT gates, and CNOTs that do not share qubits can be executed simultaneously within a single entangling layer. In the construction of Ref.~\cite{PhysRevX.15.021082} (left), the $w-1$ arrows originating from $d_1$ cannot be executed simultaneously, causing the circuit depth to scale as $O(w)$ with the Hamming weight. In our new construction (right), the 
required CNOTs can be scheduled in two entangling layers,
reducing the overall circuit depth to $O(1)$.

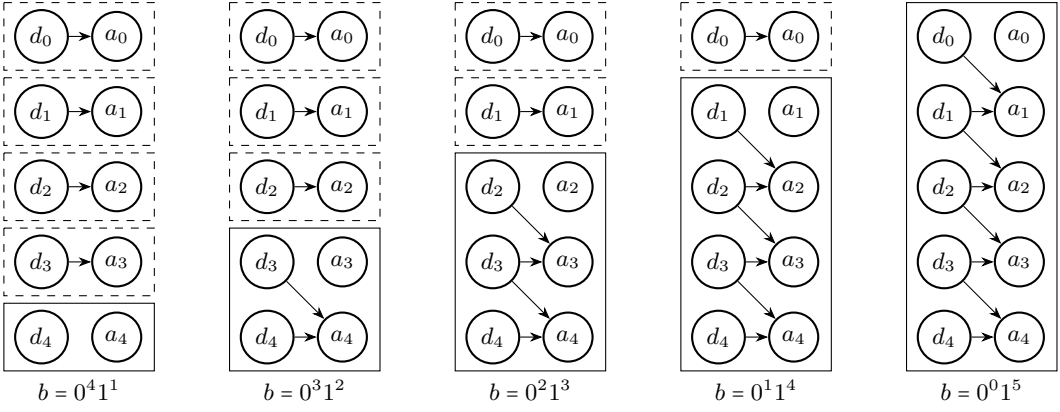
\begin{figure*}
    \centering
    \begin{tikzpicture}
        \begin{scope}[every node/.style={circle,thick,draw}]
            \node (4x0) at (0,1.7) {$d_0$};
            \node (4x1) at (0,0.7) {$d_1$};
            \node (4x2) at (0,-0.3) {$d_2$};
            \node (4x3) at (0,-1.3) {$d_3$};
            \node (4x4) at (0,-2.3) {$d_4$};
            \node (4a0) at (1,1.7) {$a_0$};
            \node (4a1) at (1,0.7) {$a_1$} ;
            \node (4a2) at (1,-0.3) {$a_2$};
            \node (4a3) at (1,-1.3) {$a_3$};
            \node (4a4) at (1,-2.3) {$a_4$};

            \node (3x0) at (3,1.7) {$d_0$};
            \node (3x1) at (3,0.7) {$d_1$};
            \node (3x2) at (3,-0.3) {$d_2$};
            \node (3x3) at (3,-1.3) {$d_3$};
            \node (3x4) at (3,-2.3) {$d_4$};
            \node (3a0) at (4,1.7) {$a_0$};
            \node (3a1) at (4,0.7) {$a_1$} ;
            \node (3a2) at (4,-0.3) {$a_2$};
            \node (3a3) at (4,-1.3) {$a_3$};
            \node (3a4) at (4,-2.3) {$a_4$};

            \node (2x0) at (6,1.7) {$d_0$};
            \node (2x1) at (6,0.7) {$d_1$};
            \node (2x2) at (6,-0.3) {$d_2$};
            \node (2x3) at (6,-1.3) {$d_3$};
            \node (2x4) at (6,-2.3) {$d_4$};
            \node (2a0) at (7,1.7) {$a_0$};
            \node (2a1) at (7,0.7) {$a_1$} ;
            \node (2a2) at (7,-0.3) {$a_2$};
            \node (2a3) at (7,-1.3) {$a_3$};
            \node (2a4) at (7,-2.3) {$a_4$};

            \node (1x0) at (9,1.7) {$d_0$};
            \node (1x1) at (9,0.7) {$d_1$};
            \node (1x2) at (9,-0.3) {$d_2$};
            \node (1x3) at (9,-1.3) {$d_3$};
            \node (1x4) at (9,-2.3) {$d_4$};
            \node (1a0) at (10,1.7) {$a_0$};
            \node (1a1) at (10,0.7) {$a_1$} ;
            \node (1a2) at (10,-0.3) {$a_2$};
            \node (1a3) at (10,-1.3) {$a_3$};
            \node (1a4) at (10,-2.3) {$a_4$};

            \node (x0) at (12,1.7) {$d_0$};
            \node (x1) at (12,0.7) {$d_1$};
            \node (x2) at (12,-0.3) {$d_2$};
            \node (x3) at (12,-1.3) {$d_3$};
            \node (x4) at (12,-2.3) {$d_4$};
            \node (a0) at (13,1.7) {$a_0$};
            \node (a1) at (13,0.7) {$a_1$} ;
            \node (a2) at (13,-0.3) {$a_2$};
            \node (a3) at (13,-1.3) {$a_3$};
            \node (a4) at (13,-2.3) {$a_4$};
        \end{scope}
        \node (eq) at (0.5,-3) {$b=0^41^1$} ;
        \node (eq) at (3.5,-3) {$b=0^31^2$} ;
        \node (eq) at (6.5,-3) {$b=0^21^3$} ;
        \node (eq) at (9.5,-3) {$b=0^11^4$} ;
        \node (eq) at (12.5,-3) {$b=0^01^5$} ;
        \begin{scope}[>={Stealth[black]}]
            \draw[->] (x0) -- (a1);
            \draw[->] (x1) -- (a2);
            \draw[->] (x2) -- (a3);
            \draw[->] (x3) -- (a4);
            \draw[->] (x1) -- (a1);
            \draw[->] (x2) -- (a2);
            \draw[->] (x3) -- (a3);
            \draw[->] (x4) -- (a4);

            \draw[->] (1x0) -- (1a0);
            \draw[->] (1x1) -- (1a2);
            \draw[->] (1x2) -- (1a3);
            \draw[->] (1x3) -- (1a4);
            \draw[->] (1x2) -- (1a2);
            \draw[->] (1x3) -- (1a3);
            \draw[->] (1x4) -- (1a4);

            \draw[->] (2x0) -- (2a0);
            \draw[->] (2x1) -- (2a1);
            \draw[->] (2x2) -- (2a3);
            \draw[->] (2x3) -- (2a4);
            \draw[->] (2x3) -- (2a3);
            \draw[->] (2x4) -- (2a4);

            \draw[->] (3x0) -- (3a0);
            \draw[->] (3x1) -- (3a1);
            \draw[->] (3x2) -- (3a2);
            \draw[->] (3x3) -- (3a4);
            \draw[->] (3x4) -- (3a4);

            \draw[->] (4x0) -- (4a0);
            \draw[->] (4x1) -- (4a1);
            \draw[->] (4x2) -- (4a2);
            \draw[->] (4x3) -- (4a3);
        \end{scope}

        \draw[dashed] (-0.5,2.15) rectangle (1.5,1.25);
        \draw[dashed] (-0.5,1.15) rectangle (1.5,0.25);
        \draw[dashed] (-0.5,0.15) rectangle (1.5,-0.75);
        \draw[dashed] (-0.5,-0.85) rectangle (1.5,-1.75);
        \draw[-] (-0.5,-1.85) rectangle (1.5,-2.75);

        \draw[dashed] (2.5,2.15) rectangle (4.5,1.25);
        \draw[dashed] (2.5,1.15) rectangle (4.5,0.25);
        \draw[dashed] (2.5,0.15) rectangle (4.5,-0.75);
        \draw[-] (2.5,-0.85) rectangle (4.5,-2.75);

        \draw[dashed] (5.5,2.15) rectangle (7.5,1.25);
        \draw[dashed] (5.5,1.15) rectangle (7.5,0.25);
        \draw[-] (5.5,0.15) rectangle (7.5,-2.75);

        \draw[dashed] (8.5,2.15) rectangle (10.5,1.25);
        \draw[-] (8.5,1.15) rectangle (10.5,-2.75);

        \draw[-] (11.5,2.15) rectangle (13.5,-2.75);
    \end{tikzpicture}
    \caption{Circuit construction underlying the compiled Simon-$5$ query circuits in the graph representation. Using the reduction procedure, we can
    trace out an appropriate subset of the leading dashed boxes to obtain Simon-$m$ for any $m$ satisfying $\HW(b)\le m<5$.
    }
    \label{fig:simon-n5-graph}
\end{figure*}

As an illustrative example, \cref{fig:simon-n5-graph} shows the circuit construction
for Simon-$5$ for hidden strings of the form $b=0^{5-i}1^i$ with $\HW(b)=i=1, 2, \dots, 5$. Each circuit uses $2n$ qubits, with $n$ data qubits $d_j$ and $n$ ancilla qubits $a_j$. The solid and dashed boxes indicate groups of qubits that share no CNOT gates across their boundaries and therefore remain unentangled with one another in the noiseless case.
Under the factorized-CPTP-map assumption of \cite{PhysRevX.15.021082} (i.e., provided crosstalk between different boxes is negligible), a
partial trace can then be applied to reduce the problem from Simon-$n$ to Simon-$m$ for any
$m$ satisfying $\HW(b)\le m < n$ in the noisy setting as well.

We conducted experiments on two IBM QPUs: a $156$-qubit device, Boston (ibm\_boston), and a $120$-qubit device, Miami (ibm\_miami).
Boston is a Heron processor and Miami is a later-generation Nighthawk processor,
with heavy-hex and square-lattice qubit connectivity graphs, respectively.
At the time of our experiments, Boston offered the lowest two-qubit gate error rates among the systems available to us. The specifications of both QPUs are summarized in \cref{app:spec}.

The compiled Simon query circuit for a problem of size $n$
uses $2n$ qubits and can be embedded on a linear nearest-neighbor layout;
in the case $w=n$, this requires a chain of length $2n-1$ plus one isolated ancilla qubit $a_0$ (see \cref{fig:simon-n5-graph}).
This linear structure plays a central role in the compilation. On a square-grid architecture,
a linear chain can be straightforwardly embedded by mapping each qubit along a path that traverses the grid,
allowing the full $120$-qubit chip of Miami to be utilized, giving $n_{\max} = 60$.
Boston's connectivity graph, on the other hand, contains a linear chain of length 129, but this chain does not span the entire chip.
Consequently, on Boston we have $n_{\max} = 65$, where $129$ qubits are arranged into a linear chain for the case $w = n$,
with one additional qubit ($a_0$) assigned independently (see \cref{fig:simon-n5-graph};
the maximum required linear chain length is $2 \min(n,w)-1$ for $w \geq 2$).

After identifying the linear chain on each QPU, we compiled the circuit onto the corresponding region of the chip using an as-late-as-possible (ALAP) scheduling strategy, which initializes each qubit immediately before its first operation. Since the pre-compiled circuit maps directly onto a linear chain, the resulting compiled circuits require no additional SWAP gates to satisfy the connectivity constraints. This also implies that the circuits are very shallow, leaving insufficient idle time for error suppression techniques such as dynamical decoupling (DD), which was used extensively to enhance performance in \cite{PhysRevX.15.021082}. An example of \wSimon{2}{3} compiled on Boston is shown in \cref{fig:timeline}.

\begin{figure}[h!]
    \centering
    \includegraphics[width=0.5\textwidth]{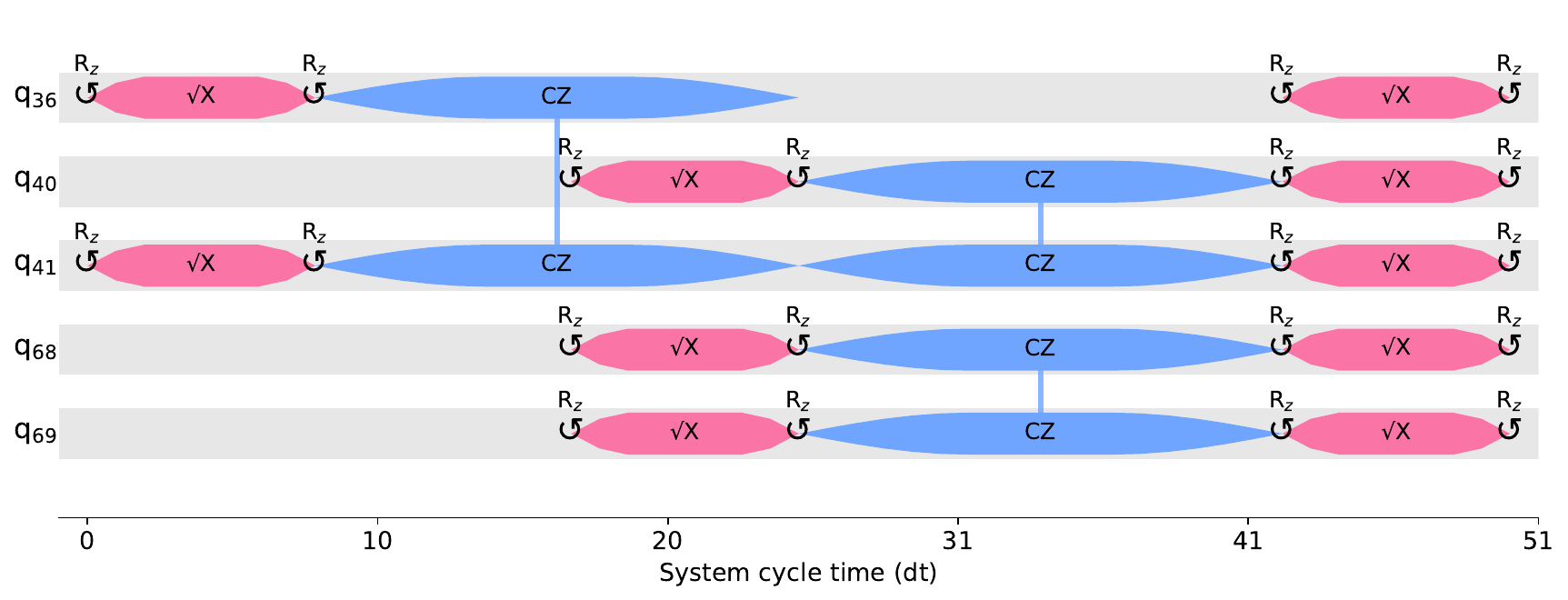}
    \caption{Circuit for \wSimon{2}{3} compiled on Boston using the ALAP scheduling strategy. The native gates shown are $R_z$ (circular arrows), $\sqrt{X}$ (pink), and CZ (blue). Note that the only idle period on $q_{36}$ is insufficient to accommodate even the shortest two-pulse DD sequence.}
    \label{fig:timeline}
\end{figure}

\section{Results and discussion}\label{sec:results}
We present two sets of results: the restricted-HW version \wSimon{w}{n} and the original Simon's problem \wSimon{n}{n}.

\begin{figure}[h!]
\centering
\includegraphics[width=0.5\textwidth]{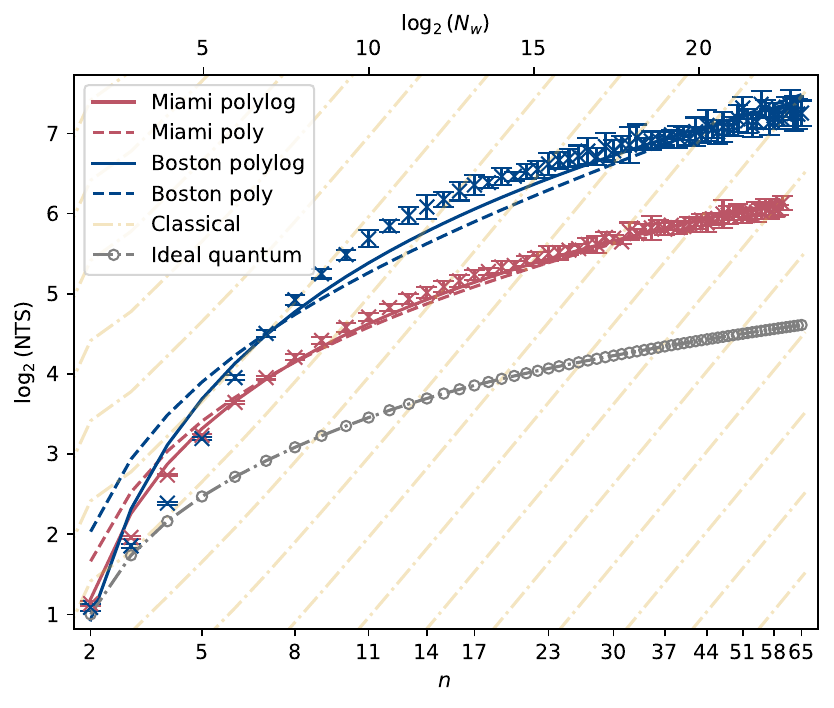}
\caption{$\log_2(\NTS_Q)$ as a function of $n$ (bottom axis) for the restricted-HW Simon problem with cutoff $w=5$ on Boston (blue, up to $n=65$) and Miami (red, up to $n=60$); the top axis shows $\log_2(N_w)$. The solid lines denote fits using the polylog model [\cref{eq:polylog-NTS}], and the dashed lines denote fits using the poly model [\cref{eq:poly-NTS}]. The error bars show bootstrapped $1\sigma$ uncertainties. The yellow dash-dot lines represent the scaling of $\log_2(\NTSClb)$ and also serve as a visual guide for the large-$n$ scaling of the poly model. The gray dashed line is $\log_2$ of the interpolated noiseless-quantum estimate given by \cref{eq:NTSIQ-interp}.}
\label{fig:main-plot-nmax}
\end{figure}

Our first result is shown in \cref{fig:main-plot-nmax}, which plots $\log_2(\NTS_Q)$ as a function of problem size $n$ for the restricted-HW Simon problem with cutoff $w=5$ for the Hamming weight,
using data up to $n=65$ on Boston and up to $n=60$ on Miami.
Note that for $n<w$ (in this case, for $n < 5$), the cutoff for the Hamming weight is effectively $w=n$ because there are no bitstrings of length $n$ with Hamming weight greater than $n$.
We used bootstrapping to compute all means and standard deviations, as documented in \cref{app:bootstrap}.
As explained in \cref{sec:o1-oracle}, for each device we extracted all $m<n$ values from the largest restricted-HW circuits in this data set,
namely those implemented at $n=65$ on Boston and $n=60$ on Miami.

\begin{figure}
    \centering
    \includegraphics[width=0.5\textwidth]{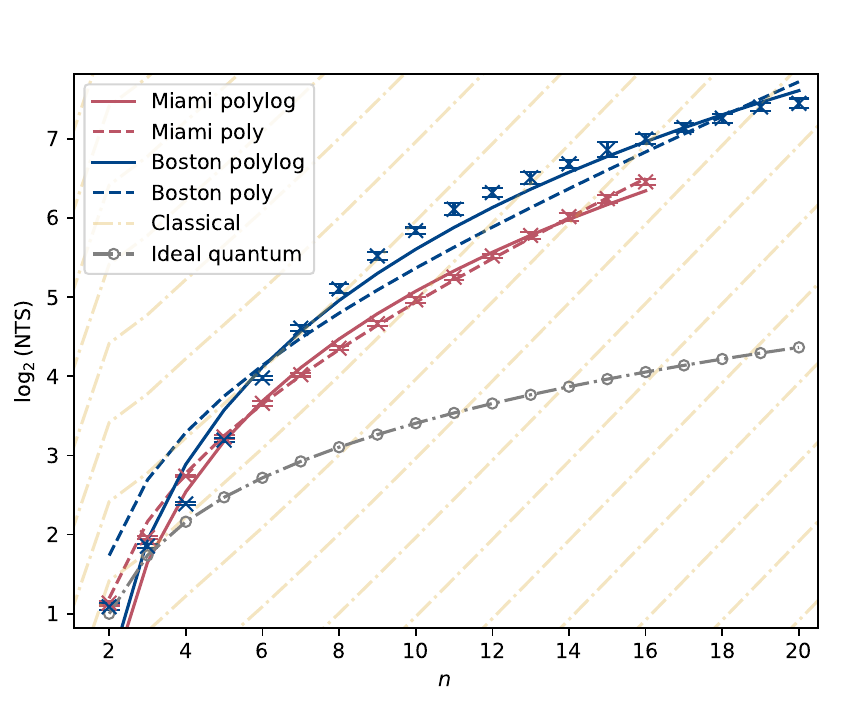}
    \caption{$\log_2(\NTS_Q)$ as a function of $n$ on Boston (blue) and Miami (red) for the original \wSimon{n}{n}, up to $n=20$ and $n=16$, respectively. The line styles and reference curves are the same as in \cref{fig:main-plot-nmax}.}
    \label{fig:main-plot-original}
\end{figure}

Our second result is shown in \cref{fig:main-plot-original},
which plots $\log_2(\NTS_Q)$ as a function of problem size $n$ for the largest contiguous unrestricted-Simon data set used in the scaling analysis,
extending up to $n=16$ on Miami and $n=20$ on Boston. An isolated usable \wSimon{18}{18} point on Miami is shown in \cref{fig:summary-range},
but because the $n=17$ unrestricted-Simon instance did not yield a usable $\NTS_Q$ value,
this point is not included in the unrestricted-Simon scaling fit.
While the unrestricted data do not extend as far as in the restricted-HW case,
this is because the classical post-processing becomes much more memory intensive when $w=n$ and the candidate set size grows as $2^n-1$;
moreover, the convergence time to a guess of $b$ exceeds the 48-hour HPC time limit imposed in our analysis for sufficiently large values of $n$.
While we expect that information set decoding algorithms
like BJMM (\cite{becker2012decoding}, \cite{narisada2024solving})
can be adapted to overcome this challenge,
this was not done in this work.
Even so, the range of Hamming weights shown exceeds that of Ref.~\cite{PhysRevX.15.021082}.

To determine the nature of the quantum speedup,
we fit the $\NTS_Q$ data to both the polylog and poly models.
Since the classical lower bound is described by the poly model, a polylog fit to $\NTS_Q$ indicates an exponential speedup over the classical benchmark.
If the poly model instead provides a better fit, the speedup is characterized as polynomial, with the degree of advantage determined by comparing the scaling exponents $\beta_Q$ and $\beta_C$;
a polynomial speedup requires $\beta_Q < \beta_C$.
\Cref{fig:main-plot-original,fig:main-plot-nmax} suggest visually that the polylog model provides a better fit than the poly model on Boston.
The distinction is less clear on Miami.

For the restricted-HW version we apply this scaling analysis to 36 experimental curves: one for each device (Miami or Boston) and $w \in [1, w_{\mathrm{max}}]$,
where $w_{\mathrm{max}} = 16$ and $w_{\mathrm{max}} = 20$ for Miami and Boston, respectively.
Each of these curves is the set of points $(n, \NTS_Q(n;w))$, where $\NTS_Q(n;w)$ is computed using \cref{eq:NTS-wn} from experimental data.
Scaling plots for $w \in [2,16]$, presented in the same format as the $w = 5$ curve in \cref{fig:main-plot-nmax}, are provided in \cref{fig:grid-plot} in \cref{app:stat-analysis}.

For the unrestricted Simon problem, we have one experimental curve for each device, with points $(n,\NTS_Q(n;n))$ and endpoint $n_{\rm end}=16$ on Miami and $n_{\rm end}=20$ on Boston,  shown in \cref{fig:main-plot-original}. 

In addition to fitting the full available contiguous range, we perform an endpoint-stability analysis in which the same unrestricted-Simon curve is refitted on the initial intervals $n=2,\dots,n_{\rm end}$, with $n_{\rm end}=7,\dots,16$ on Miami and $n_{\rm end}=7,\dots,20$ on Boston. This analysis tests whether the preferred scaling model and fitted exponent remain stable as larger unrestricted-Simon instances are included. In figures below where these unrestricted-Simon endpoint fits are plotted on the same horizontal coordinate as the restricted-HW fits, such as \cref{fig:R2-AIC,fig:poly-miami,fig:R2-AIC-qctrl}, the horizontal coordinate for the unrestricted curves is $n_{\rm end}$, while it is the Hamming-weight cutoff $w$ for the restricted-HW curves.

We start the endpoint-stability analysis at $n_{\rm end}=7$ because both the poly and polylog models have two free parameters. The fit over $n=2,\dots,n_{\rm end}$ contains $n_{\rm end}-1$ data points and hence $n_{\rm end}-3$ residual degrees of freedom. At $n_{\rm end}=6$ this leaves only three residual degrees of freedom, which is too few to reliably discriminate between two smooth two-parameter curves.

To quantitatively assess which model better describes $\NTS_Q$, we employ two fit diagnostics:
the coefficient of determination $R^2$ and the Akaike Information Criterion (AIC)~\cite{1100705}.
Because many of the data points at different $n$ are obtained by marginalizing or partially tracing data from the same largest circuits, the residuals are correlated. We therefore use $R^2$ and AIC as heuristic model-selection diagnostics rather than as fully covariance-aware hypothesis tests. The $R^2$ statistic, bounded between $0$ and $1$, measures the proportion of variance in the data explained by the model via the ratio of the residual sum of squares to the total sum of squares, with $R^2 = 1$ corresponding to a perfect fit. The AIC is derived from a likelihood model and incorporates a penalty for model complexity. Although the absolute value of the AIC is not directly interpretable, a lower AIC relative to a competing model indicates a better fit within the assumed fitting model.

\begin{figure*}
  \centering
  \includegraphics[width=0.48\textwidth]{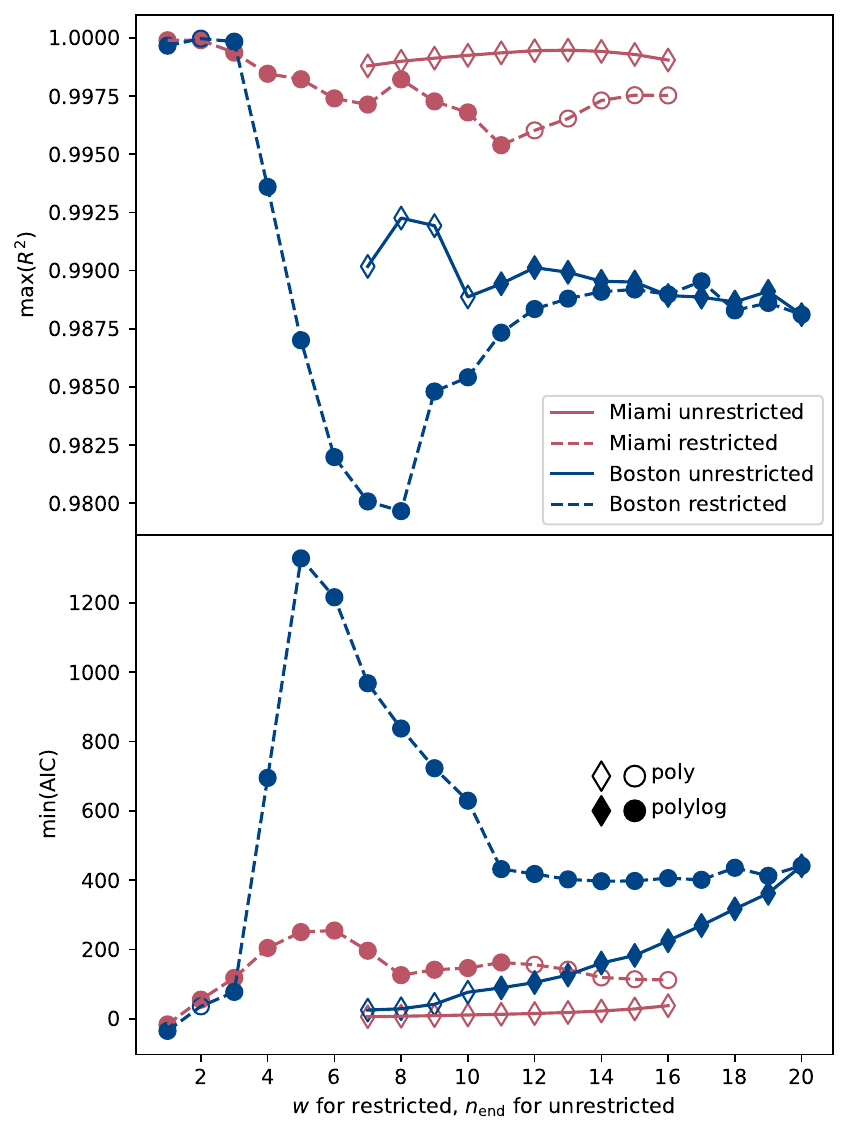}
  \includegraphics[width=0.47\textwidth]{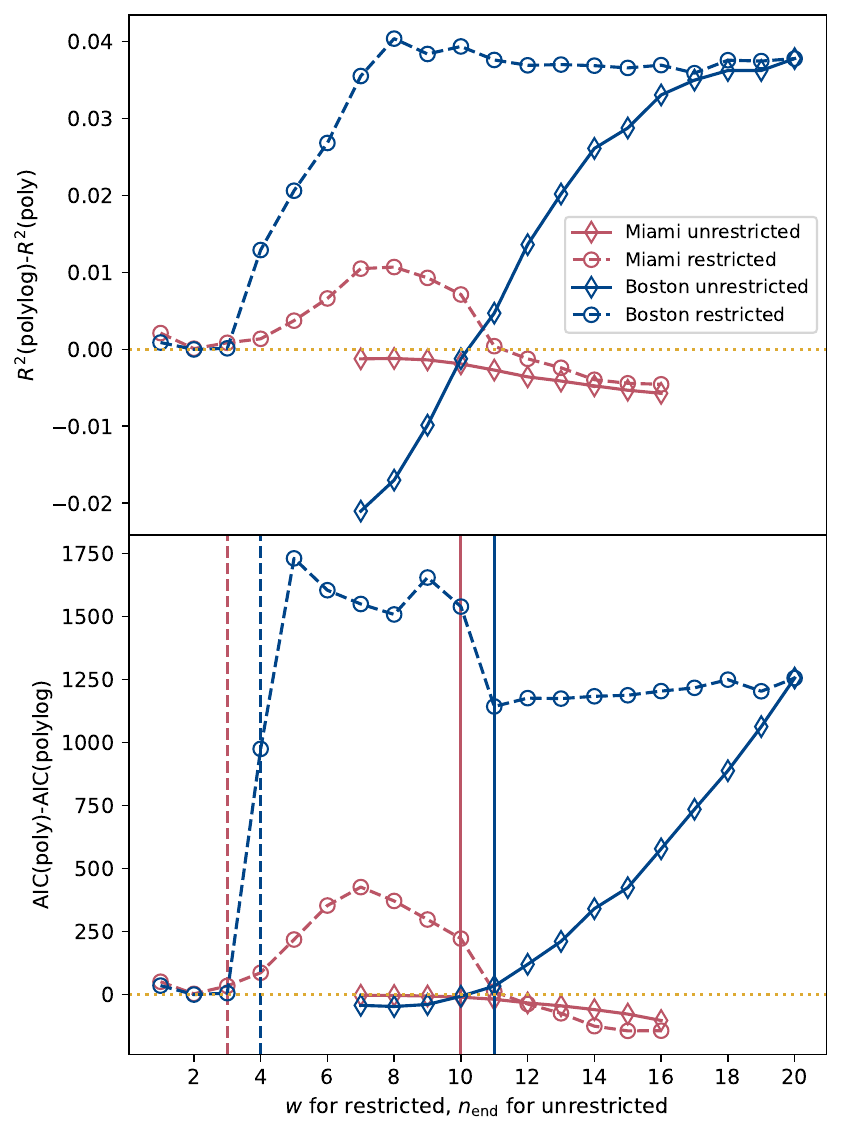}
  \caption{
Test of whether the polylog or poly model is the better fit for $\NTS_Q$.
  For both Boston (blue) and Miami (red), solid lines with diamond markers denote unrestricted Simon fits for \wSimon{n}{n}; for these curves, the horizontal coordinate is the largest problem size $n_{\rm end}$ included in the fit, and each point is obtained by fitting the unrestricted data over the  range $n=2,\dots,n_{\rm end}$.
  Dashed lines with circle markers denote restricted-HW fits for \wSimon{w}{n}; for these curves, the horizontal coordinate is the Hamming-weight cutoff $w$.
  For the restricted-HW version, $w\in[1,w_{\max}]$, with $w_{\max}=16$ on Miami and $w_{\max}=20$ on Boston; each point is fitted for $n\in[2,n_{\max}]$, where $n_{\max}$ is the largest problem size for which the bar height in \cref{fig:summary-range} is at least $w$.
  For the unrestricted version, $n_{\rm end}\in[7,16]$ on Miami and $n_{\rm end}\in[7,20]$ on Boston; the lower endpoint choice $n_{\rm end}=7$ is discussed in the main text.
  Top left: at each displayed coordinate value, the marker indicates the model with the highest $R^2$ value.
  Bottom left: at each displayed coordinate value, the marker indicates the model with the lowest AIC value.
  The $R^2$ and AIC criteria are in complete agreement across all displayed coordinate values, as confirmed by the consistent filled and open markers between the top and bottom left panels.
  Top right: difference of the $R^2$ values of the polylog and poly models.
  Bottom right: difference of the AIC values of the poly and polylog models.
  The differences are defined such that a value above zero favors the polylog model, while a value below zero favors the poly model.
  Vertical lines indicate the thresholds at which $|\Delta \text{AIC}| \ge 10$ for the strongly favored model-preference regimes summarized in the main text.
}
    \label{fig:R2-AIC}
\end{figure*}

The left panels of \cref{fig:R2-AIC} plot the $R^2$ and AIC diagnostics over the displayed fit-coordinate range;
filled markers denote the polylog model and open markers denote the poly model.
For restricted-HW fits, the fit coordinate is the Hamming-weight cutoff $w$.
For unrestricted Simon fits, the fit coordinate is the largest problem size $n_{\rm end}$ included in the fit, where each point is obtained by fitting the unrestricted data over the contiguous range $n=2,\dots,n_{\rm end}$.
The right panels of \cref{fig:R2-AIC} plot the pairwise differences at each displayed fit-coordinate value; the differences are defined such that a value above zero favors the polylog model and a value below zero favors the poly model.
The magnitude of the differences also gives a sense of how decisively one model outperforms the other.
To ensure a reliable distinction between the two models, we adopt the rule of thumb that a model is strongly favored only when $|\Delta \text{AIC}| \ge 10$; see \cref{app:AIC}.
We use this threshold only to identify regimes in which one model is strongly favored.
For restricted-HW fits, we denote the corresponding threshold by $w_c$ when the preferred model is stable beyond that threshold.
For unrestricted Simon fits, we denote the corresponding threshold by $n_{\rm end,c}$.
For Miami restricted data, the preferred model changes again at larger $w$, so we report the two strongly favored regimes separately rather than assigning a single stable-after threshold.
The vertical lines in \cref{fig:R2-AIC} indicate the thresholds of the strongly favored regimes summarized in \cref{tab:models}.

As shown in \cref{fig:R2-AIC}, the $R^2$ and AIC criteria are in complete agreement across all displayed fit-coordinate values.
On Boston, they favor the polylog model for restricted-HW fits starting at $w_c=4$ and for unrestricted Simon fits starting at $n_{\rm end,c}=11$.
On Miami, for the restricted-HW case, they favor the polylog model from $w=3$ up to $w=11$, beyond which the poly model is a better fit,
while for the unrestricted case, they favor the poly model for $n_{\rm end}\ge n_{\rm end,c}=10$.
We provide additional statistical analysis in \cref{app:stat-analysis}.
The resulting preferred models for $\NTS_Q$ on both machines are summarized in \cref{tab:models}.

\begin{table}[h!]
    \centering
    \begin{tabular}{c|c|c|c}
        \hline
        \hline
        & Type & polylog (exp. speedup) & poly (poly speedup)\\
        \hline
        \multirow{2}{*}{Boston} & restricted & $w\in[4,20]$ & - \\
        \cline{2-4}
& unrestricted & $n_{\rm end}\in[11,20]$ & - \\
        \hline
        \multirow{2}{*}{Miami} & restricted & $w\in[3,11]$ & $w\in[12,16]$ \\
        \cline{2-4}
        & unrestricted & - & $n_{\rm end}\in[10,16]$ \\
                \hline
        \hline
    \end{tabular}
    \caption{Summary of 
    the preferred scaling models for $\NTS_Q$ and the associated speedups. For restricted-HW rows, the reported ranges are Hamming-weight cutoffs $w$; for unrestricted rows, the reported ranges are the endpoints $n_{\rm end}$.}
    \label{tab:models}
\end{table}

\begin{figure}[h!]
    \centering
    \includegraphics[width=0.49\textwidth]{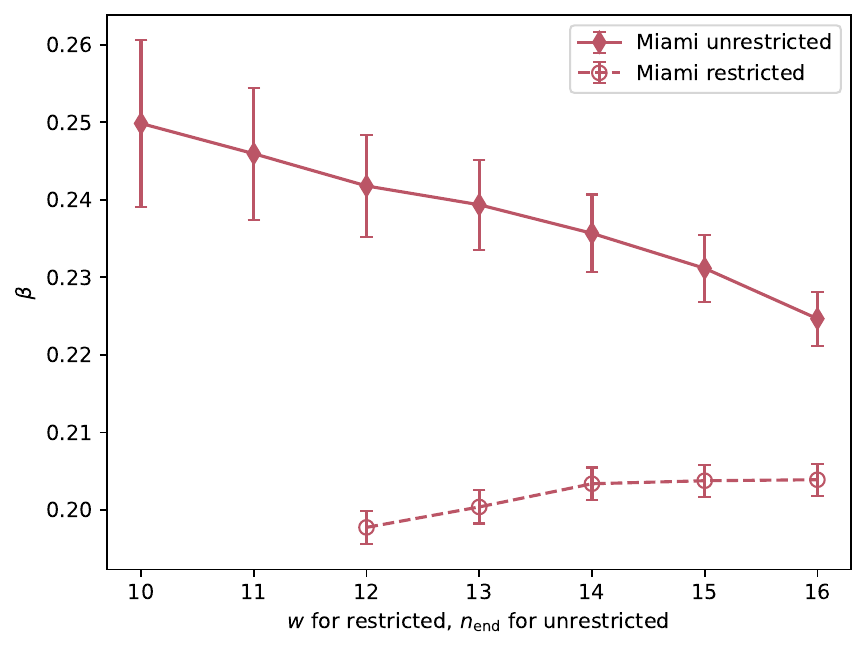}
    \caption{
The fitted scaling exponent $\beta$ of the poly model on Miami. For the unrestricted Simon problem (solid diamonds), the horizontal coordinate is the endpoint $n_{\rm end}$, shown for $n_{\rm end}\in[10,16]$. For the restricted-HW problem (circles), the horizontal coordinate is the Hamming-weight cutoff $w$, shown for $w\in[12,16]$. In both cases, the fitted quantum exponent satisfies $\beta_Q<\beta_C=1/2$ throughout the displayed range, indicating polynomial quantum speedup wherever the poly model is the preferred fit. Error bars represent one standard deviation ($1\sigma$) of the fitted parameter obtained from bootstrapped data.
    }
    \label{fig:poly-miami}
\end{figure}

Whenever the poly model is the preferred fit on Miami, one must compare its exponent with the classical value to determine whether the scaling corresponds to a speedup or a slowdown.
\Cref{fig:poly-miami} shows the fitted exponent $\beta$ over the ranges where the poly model is strongly favored ($|\Delta \text{AIC}| \ge 10$):
$n_{\rm end}\in[10,16]$ for the unrestricted problem and $w\in[12,16]$ for the restricted-HW problem.
In both cases, the classical exponent $\beta_C=1/2$ exceeds the quantum exponent $\beta_Q$, confirming a polynomial quantum speedup in those ranges, as indicated in the rightmost column of \cref{tab:models}.

For the restricted-HW data, the largest problem sizes analyzed here are $n=65$ on Boston and $n=60$ on Miami, corresponding to $130$ and $120$ qubits, respectively, before any further marginalization to smaller reduced instances. In both cases the main restricted-HW scaling plot uses cutoff $w=5$.
For the unrestricted Simon's problem, the largest reduced circuits for which $\NTS_Q$ could be evaluated involve $36$ qubits at $w=n=18$ on Miami and $40$ qubits at $w=n=20$ on Boston, as illustrated in \cref{fig:summary-range}.
For Miami, however, the unrestricted scaling analysis uses the largest contiguous range of available unrestricted-Simon points, which extends only to $n=16$, because the $n=17$ unrestricted-Simon instance did not yield a usable $\NTS_Q$ value.
The previously reported algorithmic quantum speedup for the restricted-HW Simon problem was demonstrated on $127$-qubit Sherbrooke and Brisbane QPUs~\cite{PhysRevX.15.021082}.
The present results extend that demonstration to newer QPUs, including the $156$-qubit Boston device and the $120$-qubit Miami device, broaden the Hamming-weight range over which an exponential speedup is observed, and show that the observed scaling advantage extends to the original Simon problem \wSimon{n}{n} over the finite problem-size range studied here.
The speedup trend is promising, and we expect it to persist at larger problem sizes; however, we do not present results all the way up to $n_{\max}$ because our classical post-processing algorithm (see \cref{app:algo-nisq}) becomes memory limited.

\begin{figure*}[t]
    \centering
    \includegraphics[width=\textwidth]{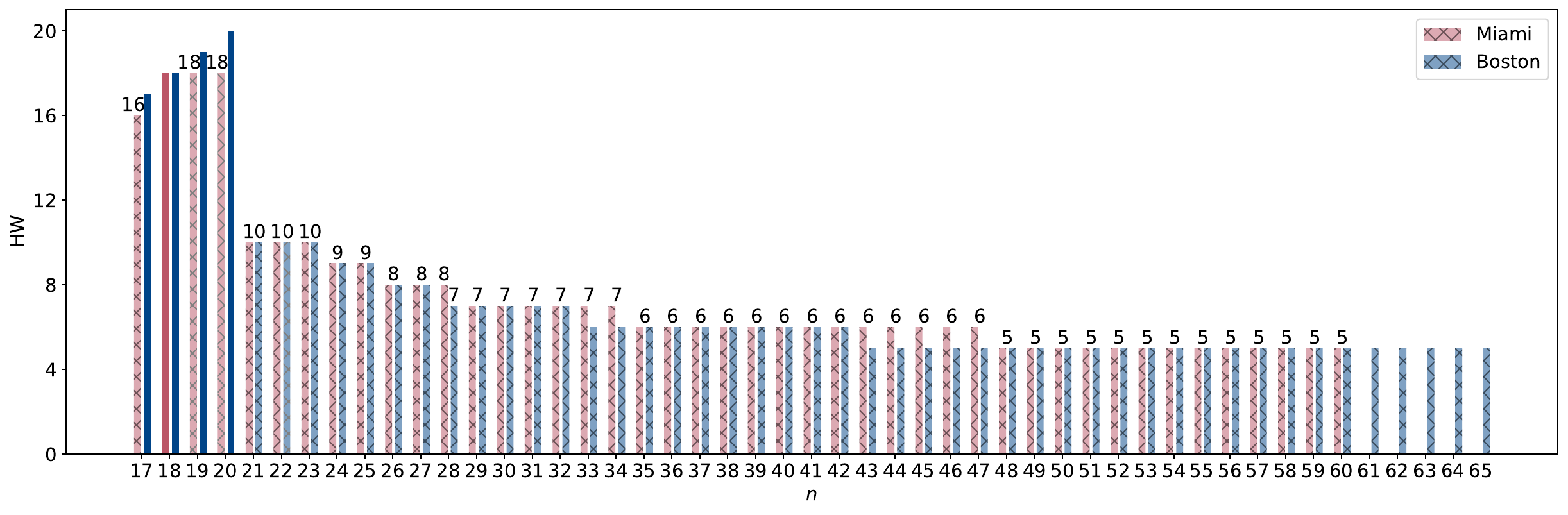}
\caption{Availability range of the NTS analysis as a function of problem size. For each problem size $n$ shown, the bar height gives the largest Hamming-weight cutoff $w_{\max}^{\rm avail}(n)$ for which $\NTS_Q$ could be evaluated on Boston (blue) and Miami (red). A solid bar reaching $w_{\max}^{\rm avail}(n)=n$ corresponds to the original Simon problem \wSimon{n}{n}; a cross-hatched bar with $w_{\max}^{\rm avail}(n)<n$ indicates that only restricted-HW instances \wSimon{w}{n} were available at that problem size. For example, on Miami the $n=17$ unrestricted-Simon instance did not yield a usable $\NTS_Q$ value, whereas the \wSimon{18}{18} instance did. Consequently, the Miami unrestricted-Simon scaling analysis uses the largest contiguous available range, ending at $n=16$, even though an isolated unrestricted-Simon point is available at $n=18$. Conversely, for a fixed cutoff $w$ included in the restricted-HW scaling analysis, the endpoint $n_{\max}(w)$ is chosen from the largest contiguous range of problem sizes for which the corresponding device's bar height is at least $w$. Missing or truncated bars occur when the NTS computation did not yield a usable value, either because the classical post-processing exceeded the 48-hour limit or because the experimental fidelity was too low for the post-processing algorithm we are using (see $\hat{f}(i) \leq 0.5$ condition in \cref{app:algo-nisq}).}
    \label{fig:summary-range}
\end{figure*}

\section{Conclusions}\label{sec:conclusion}
Demonstrating a quantum speedup that scales favorably with problem size is central to establishing the practical utility of quantum computers.
Simon's problem,
an early instance of the Abelian hidden subgroup problem and a conceptual precursor to Shor's factoring algorithm,
provides a natural setting for this pursuit:
it requires an exponential number of oracle queries classically
yet only a linear number of oracle queries on a noiseless quantum computer.
Thus, in the oracle-query model, Simon's problem exhibits a theoretical quantum speedup when the resources required to implement the oracle are not accounted for.

A restricted-HW version of Simon's problem,
in which the hidden bitstring is restricted by $\HW(b)\le w\le n$,
was studied in Ref.~\cite{PhysRevX.15.021082},
where it was demonstrated that polylogarithmic scaling of $\NTS_Q$ in $N_w$, which corresponds to an exponential speedup in the $N_w$-based NTS metric used in that work, is achievable on NISQ devices when $w$ is fixed as $n$ grows.
For fixed $w$, the classical lower bound scales as $n^{w/2}$;
in the unrestricted case $w=n$,
the classical query complexity is instead $\Omega(2^{n/2})$.

In this work, we revisited this problem with a new circuit compilation scheme
that achieves constant circuit depth.
We performed Simon's algorithm experiments on the IBM Quantum platform
and demonstrated that the $156$-qubit Boston device exhibits exponential quantum speedup
over the Hamming-weight range for which we could reliably distinguish
between poly and polylog scaling,
while the $120$-qubit Miami device similarly exhibits exponential speedup
at low-to-intermediate Hamming weights and polynomial speedup at higher Hamming weights.
Moreover, over the finite range of unrestricted-HW instances for which our NTS computation was feasible, the observed scaling advantage also appears for the original Simon problem with $w=n$, and is therefore not limited to the small-$w$ regime.
The observed performance is consistent with the reduced depth
and routing overhead of the new 
compiled-circuit construction,
although comparisons with \cite{PhysRevX.15.021082}
should be interpreted with some care because
the present experiments were performed on newer hardware.

These results significantly extend experimental demonstrations
of algorithmic quantum speedup in the oracle model.
In particular, relative to our previous work \cite{PhysRevX.15.021082}, the main advances are
  the broader Hamming-weight range over which exponential speedup is observed,
  the observation of polynomial speedup on Miami when exponential scaling is no longer favored,
  and the extension of the experiments to a regime where the original Simon problem \wSimon{n}{n} is recovered for the problem sizes studied.
Over the problem sizes and Hamming-weight ranges studied here, we therefore observe either exponential or polynomial quantum speedup. More generally, our work adds to the experimental evidence from previous demonstrations \cite{pokharel2022demonstration,PhysRevX.15.021082} that algorithmic quantum speedup for structured, nontrivial oracle problems is becoming experimentally accessible on present-day superconducting hardware. Our constant-depth compilation scheme for Simon query circuits is also of independent interest and suggests that further hardware-aware depth reductions may be possible in other quantum algorithms, including beyond the oracle setting.

\section{Acknowledgments}
This research was conducted using IBM Quantum Systems provided through USC's IBM Quantum Innovation Center.
This research is based upon work supported by, or in part by, the U.S. Army Research
Laboratory and the U.S. Army Research Office under
Contract/Grant No. W911NF2310255, and by the Office of Naval Research under Contract/Grant No. N0001-4-26-12092.

\appendix

\section{Compiler}
\label{app:rules:oracle}
To establish a well-defined criterion for quantum speedup, and to avoid ambiguities such as how to account for post-processing resources or how much circuit optimization is permitted, we adopt the compiler framework of \cite{PhysRevX.15.021082}, which performs the following functions:
\begin{enumerate}
\renewcommand{\labelenumi}{(\roman{enumi})}
    \item it hides the implementation details of $f$ from the player;
    \item it takes a circuit $C$ with zero or more boxes labeled ``$\mcO$'' and produces
a circuit $C'$ obtained from $C$ by replacing each $\mcO$ box with
a circuit implementing it for the current $f$;
\item it further compiles $C'$ to ensure that it is compatible with the
gate set and connectivity of the NISQ device,
to reduce the number of gates and the circuit depth,
and to select the best layout of the qubits on the QPU
(e.g., avoiding the noisiest qubits and couplings),
yielding a new circuit $C''$ and classical post-processing instructions;
\item it sends $C''$ to the NISQ device for execution,
obtains the result, performs the post-processing, and returns the final result to the player.
\end{enumerate}
Note that in (ii), ``$\mcO$'' is just a box labeling a place to insert the oracle
(which is unknown to the player),
as opposed to $\mcO_f$, which is the unitary implementing the actual oracle.
The specific rules for the compilation (iii) were carefully chosen in \cite{PhysRevX.15.021082}
to ensure that we give enough freedom to the compiler,
so that it is able to remove the exponential depth
present in the original oracle implementation ($\mcO_f$),
while we do not give it so much freedom that it is able to get rid of all
quantum resources or of all errors appearing from running the circuit on a NISQ device.
We follow the same rules (1)--(4) described in \cite[Appendix A]{PhysRevX.15.021082}
for the step (iii) above.
Another freedom the compiler has is to pick a mapping between
the abstract qubits in the circuit and the physical qubits on the QPU.

\section{Construction underlying the constant-depth compilation}
\label{app:simon-oracle}
For a given $b$ we need to construct a circuit $\mcO_f$
that implements a 2-to-1 function $f$ satisfying the Simon's problem condition
with the given hidden string $b$.
Specifically,
\begin{equation}
  \mcO_f \ket{x}\ket{a} = \ket{x} \ket{a \oplus f(x)}
  \label{eq:mcOf-def}
\end{equation}
for all $x\in\{0,1\}^n$ and $a\in\{0,1\}^n$.
In order to ensure that the assumptions of \cref{thm:cl-lower-bound} establishing the classical
lower bound hold, it is important that $f$ is picked uniformly at random from
the set of all such functions.
This necessarily implies that 
a generic circuit for $\mcO_f$
has size that is, on average, exponential in $n$.
As explained in \cite[Appendix A]{PhysRevX.15.021082},
any such function $f$ can be written in the form $f = f' \circ f_b$,
where $f_b$ is any fixed 2-to-1 function satisfying
the Simon's problem condition for the given $b$,
and $f'$ is a random permutation on $\{0,1\}^n$.
Moreover, when the Simon's problem circuit, which our NISQ player uses
(as opposed to some other circuit a NISQ player can consider),
is given to the compiler, the conditions (1)--(4) for the compiler
allow it to defer $f'$ to the postprocessing of the ancilla qubits.
While the compiler needs to perform $f'$ in postprocessing every time
it receives the measurement results from the quantum computer,
our NISQ player does not actually use the measurement results of the ancilla qubits,
hence the result of this postprocessing is immediately discarded by the NISQ player.
Therefore, in order to measure the performance of the NISQ player,
we do not need to implement $f'$ at all.

It remains to pick $f_b$ and implement $\mcO_{f_b}$.
As discussed, we can pick any $f_b$ satisfying the 2-to-1 condition for the given $b$.
That is, $\forall x,y \in \{0,1\}^n$, $f_b(x) = f_b(y)$
if and only if $x = y$ or $x = y \oplus b$
for a hidden bitstring $b\in \{0,1\}^n \setminus \{0^n\}$.
Due to the symmetry of the Simon's problem circuit and the freedom of choosing the mapping
between the abstract qubits and the physical qubits, it is sufficient to
consider a single $b$ for each Hamming weight $i = \HW(b)$
(since other values of $b$ can be obtained by relabeling the bit positions).
From each such Hamming-weight class we can pick the canonical representative $b=0^{n-i}1^i$.
In the experiments, this canonical representative is used for each Hamming-weight class.
The  circuit family presented below and the $O(\HW(b))$-depth construction
presented in \cite[Appendix B]{PhysRevX.15.021082}
both satisfy this condition.

To aid in visualization, we will interchangeably use quantum circuits and directed graphs to represent the oracle. \cref{fig:circ-graph} shows an example of how to transform between the two representations. The data ($d_j$) and ancilla ($a_j$) qubits are drawn in the first and second columns in the graph representation, respectively.

\begin{figure}[h!]
    \centering
    \begin{tikzpicture}
        \begin{yquantgroup}
            \registers{
                qubit {$\ket{d_0}$} x0;
                qubit {$\ket{d_1}$} x1;
                qubit {$\ket{a_0}$} a0;
                qubit {$\ket{a_1}$} a1;
            }
            \circuit{
                cnot a0 | x0;
                cnot a0 | x1;
                cnot a1 | x1;
            }
            \equals[$\Leftrightarrow$]
        \end{yquantgroup}
        \begin{scope}[every node/.style={circle,thick,draw}]
            \node (x0) at (3.6,-0.25) {$d_0$};
            \node (x1) at (3.6,-1.25) {$d_1$};
            \node (a0) at (4.6,-0.25) {$a_0$};
            \node (a1) at (4.6,-1.25) {$a_1$};
        \end{scope}
        \begin{scope}[>={Stealth[black]}]
            \draw[->] (x0) -- (a0);
            \draw[->] (x1) -- (a0);
            \draw[->] (x1) -- (a1);
        \end{scope}
    \end{tikzpicture}
    \caption{Conversion between a pure-CNOT circuit and its graph representation. A CNOT gate is an arrow from the control qubit to the target qubit.}
    \label{fig:circ-graph}
\end{figure}
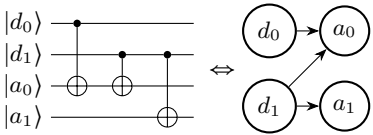

Next, we show how to construct a 2-to-1 oracle.
We represent a length-$n$ bitstring $x=x_0\cdots x_{n-1}$, where $x_j\in\{0,1\}$.
The oracle transforms $\ket{x}\ket{a}\rightarrow \ket{x}\ket{a \oplus f_b(x)}$,
where, for $b=0^{n-i}1^i$,
the function $f_b\colon \{0,1\}^n\to \{0,1\}^n$ is defined componentwise by:
\begin{equation}
[f_b(x)]_j=
\begin{cases}
x_j, & 0\le j\le n-i-1,\\
0, & j=n-i,\\
x_{j-1}\oplus x_j, & n-i+1\le j\le n-1,
\end{cases}
    \label{eq:fbx}
\end{equation}
for $1\le i\le n$.

Equivalently, we can write this transformation as
\begin{multline}
  \ket{x_0\cdots x_{n-1}}_d\ket{a_0\cdots a_{n-1}}_a \mapsto \\
  \ket{x_0\cdots x_{n-1}}_d\ket{[a \oplus f_b(x)]_0\cdots [a \oplus f_b(x)]_{n-1}}_a,
    \label{eq:fbx-state}
\end{multline}
where $[f_b(x)]_j$ is defined in \cref{eq:fbx}.

The following lemma shows that $f_b$ satisfies the 2-to-1 Simon's problem condition.

\begin{mylemma}
\label{lem:2to1}
$f_b(x)=f_b(y)$ if and only if $x= y$ or $y =  x \oplus b$, where $b=0^{n-i}1^{i}$, $1\le i\le n$. That is, $f_b$ is a 2-to-1 function.
\end{mylemma}

\begin{proof}
Let $\delta_j:=x_j\oplus y_j$. From $f_b(x)=f_b(y)$ and \cref{eq:fbx}, we obtain
\beq
x_j=y_j,\qquad j=0,\dots,n-i-1,
\eeq
hence $\delta_j=0$ for $j=0,\dots,n-i-1$, and
\beq
x_{j-1}\oplus x_j = y_{j-1}\oplus y_j,\qquad j=n-i+1,\dots,n-1.
\eeq
The latter implies $\delta_{j-1}=\delta_j$ for $j=n-i+1,\dots,n-1$, so
\beq
\delta_{n-i}=\delta_{n-i+1}=\cdots=\delta_{n-1}\in\{0,1\}.
\eeq
Therefore either all these bits are $0$, in which case $x=y$, or all these bits are $1$, in which case
\beq
y=x\oplus 0^{n-i}1^i=x\oplus b.
\eeq

Conversely, if $x=y$ then trivially $f_b(x)=f_b(y)$. If $y=x\oplus b$, then $y_j=x_j$ for $j<n-i$, the output bit with index $n-i$ is $0$ for both inputs, and for $j=n-i+1,\dots,n-1$,
\beq
y_{j-1}\oplus y_j=(x_{j-1}\oplus 1)\oplus(x_j\oplus 1)=x_{j-1}\oplus x_j.
\eeq
Hence $f_b(x)=f_b(y)$.

Thus $f_b(x)=f_b(y)$ if and only if $x=y$ or $y=x\oplus b$, so $f_b$ is a $2$-to-$1$ Simon oracle.
\end{proof}

Having shown that $f_b$ as defined in \cref{eq:fbx} and written explicitly in \cref{eq:fbx-state}
is a 2-to-1 function satisfying the Simon's problem condition for $b=0^{n-i}1^i$,
we proceed to construct $\mcO_{f_b}$ using CNOT gates.
To construct the first $n-i$ positions in \cref{eq:fbx-state}, one needs to apply the map $a_j\mapsto a_j\oplus d_j$ for $j=0,\dots,n-i-1$.
The quantum operation for this is a CNOT
from each qubit $d_j$ of the data register (control qubit)
to the corresponding qubit $a_j$ in the ancilla register (target qubit).
The qubit $a_{n-i}$ is unchanged by the oracle.
The remaining $i-1$ trailing positions are implemented by pairs of CNOTs
acting on $a_j$ (target qubit) and controlled by $d_{j-1}$ and $d_j$ for $j=n-i+1,\dots,n-1$.
\Cref{fig:combine} shows $\mcO_{f_b}$ for Simon-3 with $b=011$.
\Cref{fig:simon-n5-graph} shows $\mcO_{f_b}$ for Simon-5 in the graph representation for all $5$ possible values of $i = \HW(b)$.

The CNOT gates in this construction can always be scheduled in two entangling layers. One valid schedule is as follows. The first layer contains the CNOTs $d_j\to a_j$ for $0\le j\le n-i-1$ together with the CNOTs $d_{j-1}\to a_j$ for $n-i+1\le j\le n-1$. The second layer contains the CNOTs $d_j\to a_j$ for $n-i+1\le j\le n-1$. Within each layer, no data qubit and no ancilla qubit appears in more than one CNOT. Therefore, under the non-overlapping-qubit scheduling rule used throughout this work, the oracle $\mcO_{f_b}$ has two entangling layers for every $1\le i\le n$. The surrounding Hadamard, initialization, measurement, and native-basis decompositions add only constant depth, so the compiled Simon query circuit has $O(1)$ quantum depth.

\begin{figure}[t]
    \centering
    \begin{tikzpicture}
       \begin{yquantgroup}[operator/separation=0.8mm]
            \registers{
                qubit {} x0;
                qubit {} x1;
                qubit {} x2;
                qubit {} a0;
                qubit {} a1;
                qubit {} a2;
            }
            \circuit{
                init {$\ket{d_0}$} x0;
                init {$\ket{d_1}$} x1;
                init {$\ket{d_2}$} x2;
                init {$\ket{a_0}$} a0;
                init {$\ket{a_1}$} a1;
                init {$\ket{a_2}$} a2;
                cnot a0 | x0;
            }
            \equals[$+$]
            \circuit{
                cnot a2 | x1;
                cnot a2 | x2;
            }
            \equals
            \circuit{
                cnot a0 | x0;
                cnot a2 | x1;
                cnot a2 | x2;
            }
            \equals
        \end{yquantgroup}
        \begin{scope}[every node/.style={circle,thick,draw}]
            \node (x0) at (6.9,-0.25) {$d_0$};
            \node (x1) at (6.9,-1.25) {$d_1$};
            \node (x2) at (6.9,-2.25) {$d_2$};
            \node (a0) at (7.9,-0.25) {$a_0$};
            \node (a1) at (7.9,-1.25) {$a_1$};
            \node (a2) at (7.9,-2.25) {$a_2$};
        \end{scope}
        \begin{scope}[>={Stealth[black]}]
            \draw[->] (x0) -- (a0);
            \draw[->] (x1) -- (a2);
            \draw[->] (x2) -- (a2);
        \end{scope}
        \draw[dashed] (6.45,0.17) rectangle (8.35,-0.70);
        \draw[-] (6.45,-0.81) rectangle (8.35,-2.67);
    \end{tikzpicture}
\caption{Construction of $\mcO_{f_b}$ for Simon-$3$ problem with $b=011$ according to \cref{eq:fbx-state}.
Here ``$+$'' on the left-hand side indicates
the composition of two parts of the implementation:
the first part is responsible for the first $n-i$ ancilla qubits ($0$s in $b$),
shown using the dashed box on the right-hand side, and is implemented using one CNOT per ancilla;
the second part is responsible for the last $i-1$ ancilla qubits
(the $1$s in $b$),
shown using the solid box on the right-hand side,
and is implemented using two CNOTs per ancilla.
Qubit $a_{n-i}$ is idle ($a_1$ in this example).}
    \label{fig:combine}
\end{figure}
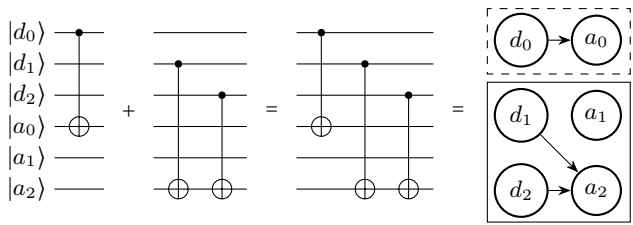

\section{QPU specifications}\label{app:spec}
We performed our demonstrations using the $156$-qubit Boston device (\texttt{ibm\_boston}) on January 15, 2026, and the $120$-qubit Miami device (\texttt{ibm\_miami}) on January 16, 2026. \cref{tab:device-spec} summarizes the QPU specifications on the respective days of the experiments, 
reporting the minimum, mean, and maximum over the qubits and couplers used. This information is extracted from the \texttt{BackendProperties} object of the QPU recorded at the time the experiment was executed.
A total of $15,000$ shots were collected for each demonstration on each QPU.

\begin{table}[h!]
    \centering
   \begin{tabular}{l|ccc|ccc}
\hline
\hline
 & \multicolumn{3}{c|}{Boston} & \multicolumn{3}{c}{Miami} \\
 \hline
 & Min & Mean & Max & Min & Mean & Max \\
\hline
$T_1(\mu s)$ & 12.43 & 267.8 & 393.6 & 14.73 & 336.4 & 562.3 \\
$T_2(\mu s)$ & 20.37 & 325.9 & 643.4 & 27.48 & 269.7 & 616.9 \\
RO err (\%) & 0.1099 & 0.9117 & 19.29 & 0.4639 & 2.649 & 34.4 \\
RO dur (ns) & 2200 & 2200 & 2200 & 2400 & 2400 & 2400 \\
1QG err (\%) & 0.005857 & 1.557 & 100 & 0.006412 & 0.03182 & 0.5476 \\
1QG dur (ns) & 32 & 32 & 32 & 32 & 32 & 32 \\
2QG err (\%) & 0.06606 & 0.465 & 14.24 & 0.1024 & 0.7062 & 15.23 \\
2QG dur (ns) & 68 & 69.29 & 108 & 68 & 148.5 & 508 \\
\hline
\hline
\end{tabular}
    \caption{QPU specifications for Boston (January 15, 2026) and Miami (January 16, 2026). Here, 1QG and 2QG denote single-qubit and two-qubit gates, respectively, and RO denotes readout, with err and dur referring to error rate and duration, respectively.
These data were obtained from the \texttt{service.job(job\_id).properties()} \texttt{BackendProperties} object via Qiskit.
    The Min, Mean, and Max columns are computed over the qubits and couplers used in the experiments.
  }
    \label{tab:device-spec}
\end{table}

On the Boston device, the maximum reported single-qubit gate error is $100\%$, arising from the calibrated $\sqrt{X}$ error entries for qubits $85$ and $146$ at the time of the experiment. We retained these qubits in the main analysis because excluding them would restrict the accessible restricted-HW data to $w=1,2,3$ and would reduce $n_{\max}$ by two. The NTS results reported in the main text are computed from the observed circuit outputs, but the presence of these calibration outliers should be kept in mind when interpreting the Boston device-specification table.

\section{Bootstrapping}\label{app:bootstrap}
We estimated uncertainties via $15,000$ bootstrap resamples of the NTS data. The available data were partitioned into 10 folds, and each bootstrap sample was formed by resampling 9 of the 10 folds with replacement. From the resulting bootstrap distribution, we extracted the mean and standard deviation of each plotted NTS value. In plots such as \cref{fig:main-plot-nmax}, the mean is shown as the data point and the error bars indicate $\pm 1\sigma$ from the bootstrap distribution.

These bootstrap means and standard deviations were then used in Mathematica to fit the parameters $\alpha$ and $\beta$, with the bootstrap standard deviation taken as the uncertainty of each data point.

\section{Statistical comparison of
the poly and polylog models}\label{app:stat-analysis}

In the main text, we focused on the $R^2$ and AIC as statistical measures for deciding between the polylog and poly models. Because the data for different $n$ values are extracted from the same $n_{\max}$ circuits, the resulting points are correlated across $n$; accordingly, these measures should be interpreted as heuristic model-selection tools rather than as fully covariance-aware hypothesis tests. 
This section presents a more comprehensive statistical analysis comparing the performance of the poly and polylog models across both QPUs and fit-coordinate values. For restricted-HW data, the fit coordinate is the Hamming-weight cutoff $w$. For unrestricted data, the fit coordinate is the endpoint $n_{\rm end}$ of the fit over $n=2,\dots,n_{\rm end}$.  

The quantities reported are the $p$-value and $t$-statistic of the fitted scaling parameter, together with the fit diagnostics Akaike Information Criterion (AIC), corrected AIC (AICc), Bayesian Information Criterion (BIC), coefficient of determination ($R^2$), and adjusted $R^2$ ($\mathrm{Adj} R^2$), all reported for the fit-coordinate values shown in \cref{fig:R2-AIC} without restricting the parameter range.
These quantities are computed using Mathematica's \texttt{NonlinearModelFit}, with confidence intervals obtained via bootstrapping as described in \cref{app:bootstrap}. We use AIC, AICc, BIC, $R^2$, and adjusted $R^2$ as model-comparison diagnostics. The $p$-value and $t$-statistic are reported only as supplementary indicators of parameter significance and should not be interpreted as independent model-selection criteria. The preferred direction is indicated adjacent to each model label: ($\Uparrow$) denotes that higher values indicate a better fit, and ($\Downarrow$) denotes that lower values indicate a better fit. Taken together, they provide a comprehensive statistical assessment of the model fits.

\begin{table*}[t]
\centering
\begin{tabular}{l|cc|cc|cc|cc|cc|cc|cc}
\hline
 $n_{\rm end}$& \multicolumn{2}{c|}{Adj $R^2$ ($\Uparrow$)} & \multicolumn{2}{c|}{$R^2$ ($\Uparrow$)} & \multicolumn{2}{c|}{AIC ($\Downarrow$)} & \multicolumn{2}{c|}{AICc ($\Downarrow$)} & \multicolumn{2}{c|}{BIC ($\Downarrow$)} & \multicolumn{2}{c|}{$t$-statistic ($\Uparrow$)} & \multicolumn{2}{c}{$p$-value ($\Downarrow$)} \\
 \hline
 & Polylog & Poly & Polylog & Poly & Polylog & Poly & Polylog & Poly & Polylog & Poly & Polylog & Poly & Polylog & Poly \\
\hline
7 & 0.9692 & \textbf{0.9902} & 0.9794 & \textbf{0.9935} & 68.23 & \textbf{25.19} & 80.23 & \textbf{37.19} & 67.6 & \textbf{24.57} & \textbf{31.84} & 23.98 & \textbf{5.802e-06} & 1.794e-05 \\
8 & 0.9753 & \textbf{0.9923} & 0.9823 & \textbf{0.9945} & 76.64 & \textbf{28.97} & 84.64 & \textbf{36.97} & 76.48 & \textbf{28.81} & \textbf{40.23} & 31.56 & \textbf{1.788e-07} & 5.995e-07 \\
9 & 0.9821 & \textbf{0.9919} & 0.9866 & \textbf{0.9939} & 81.46 & \textbf{41.88} & 87.46 & \textbf{47.88} & 81.7 & \textbf{42.12} & \textbf{52.64} & 41.96 & \textbf{3.155e-09} & 1.225e-08 \\
10 & 0.9876 & \textbf{0.9889} & 0.9904 & \textbf{0.9913} & 83.82 & \textbf{76.67} & 88.62 & \textbf{81.47} & 84.41 & \textbf{77.26} & \textbf{66.28} & 52.24 & \textbf{4.678e-11} & 2.466e-10 \\
11 & \textbf{0.9894} & 0.9847 & \textbf{0.9915} & 0.9878 & \textbf{89.39} & 122.6 & \textbf{93.39} & 126.6 & \textbf{90.29} & 123.5 & \textbf{73.69} & 57.61 & \textbf{1.281e-12} & 9.156e-12 \\
12 & \textbf{0.9901} & 0.9765 & \textbf{0.9919} & 0.9808 & \textbf{104.7} & 224.4 & \textbf{108.1} & 227.8 & \textbf{105.9} & 225.6 & \textbf{84.07} & 64.33 & \textbf{2.416e-14} & 2.677e-13 \\
13 & \textbf{0.9899} & 0.9697 & \textbf{0.9916} & 0.9748 & \textbf{126.4} & 337 & \textbf{129.4} & 340 & \textbf{127.8} & 338.5 & \textbf{92.81} & 69.5 & \textbf{5.164e-16} & 9.271e-15 \\
14 & \textbf{0.9895} & 0.9634 & \textbf{0.9912} & 0.9691 & \textbf{160.8} & 501.5 & \textbf{163.5} & 504.1 & \textbf{162.5} & 503.2 & \textbf{105} & 76.33 & \textbf{7.339e-18} & 2.428e-16 \\
15 & \textbf{0.9895} & 0.9607 & \textbf{0.991} & 0.9664 & \textbf{182.8} & 606.6 & \textbf{185.2} & 609 & \textbf{184.8} & 608.5 & \textbf{111.7} & 80.26 & \textbf{1.767e-19} & 9.324e-18 \\
16 & \textbf{0.9889} & 0.9559 & \textbf{0.9904} & 0.9618 & \textbf{225.3} & 803.3 & \textbf{227.5} & 805.5 & \textbf{227.5} & 805.4 & \textbf{122.4} & 85.99 & \textbf{2.702e-21} & 2.659e-19 \\
17 & \textbf{0.9889} & 0.9539 & \textbf{0.9903} & 0.9596 & \textbf{269.3} & 1004 & \textbf{271.3} & 1006 & \textbf{271.6} & 1007 & \textbf{133.9} & 92.1 & \textbf{3.695e-23} & 6.913e-21 \\
18 & \textbf{0.9887} & 0.9524 & \textbf{0.99} & 0.958 & \textbf{317.1} & 1205 & \textbf{318.9} & 1207 & \textbf{319.6} & 1208 & \textbf{143.6} & 96.94 & \textbf{5.883e-25} & 2.113e-22 \\
19 & \textbf{0.9891} & 0.9529 & \textbf{0.9903} & 0.9581 & \textbf{362} & 1425 & \textbf{363.7} & 1427 & \textbf{364.7} & 1428 & \textbf{155} & 103.1 & \textbf{7.555e-27} & 5.09e-24 \\
20 & \textbf{0.9881} & 0.9503 & \textbf{0.9894} & 0.9556 & \textbf{441.5} & 1697 & \textbf{443.1} & 1699 & \textbf{444.3} & 1700 & \textbf{163.2} & 106.2 & \textbf{1.318e-28} & 1.938e-25 \\
\hline
\end{tabular}
\caption{Statistical measures for Boston (unrestricted endpoint fits). Boldface indicates the favored model. For AIC, AICc, and BIC, smaller values indicate a better fit; for $R^2$ and adjusted $R^2$, larger values indicate a better fit. The $p$-value and $t$-statistic columns are reported as supplementary diagnostics for the fitted parameter and are not used as primary model-selection criteria.}
\label{tab:boston-stat}
\end{table*}

\begin{table*}[t]
\centering
\begin{tabular}{l|cc|cc|cc|cc|cc|cc|cc}
\hline
 $w$& \multicolumn{2}{c|}{Adj $R^2$ ($\Uparrow$)} & \multicolumn{2}{c|}{$R^2$ ($\Uparrow$)} & \multicolumn{2}{c|}{AIC ($\Downarrow$)} & \multicolumn{2}{c|}{AICc ($\Downarrow$)} & \multicolumn{2}{c|}{BIC ($\Downarrow$)} & \multicolumn{2}{c|}{$t$-statistic ($\Uparrow$)} & \multicolumn{2}{c}{$p$-value ($\Downarrow$)} \\
 \hline
 & Polylog & Poly & Polylog & Poly & Polylog & Poly & Polylog & Poly & Polylog & Poly & Polylog & Poly & Polylog & Poly \\
\hline
1 & \textbf{0.9997} & 0.9988 & \textbf{0.9997} & 0.9988 & \textbf{-34.82} & 0.4555 & \textbf{-34.42} & 0.8555 & \textbf{-28.34} & 6.932 & 33.6 & \textbf{202.5} & 1.69e-41 & \textbf{3.537e-89} \\
2 & 1 & \textbf{1} & 1 & \textbf{1} & 36.24 & \textbf{35.94} & 36.64 & \textbf{36.34} & 42.72 & \textbf{42.42} & \textbf{52.59} & 1.71 & \textbf{3.796e-53} & 0.09227 \\
3 & \textbf{0.9998} & 0.9997 & \textbf{0.9998} & 0.9997 & \textbf{77.57} & 82.85 & \textbf{77.97} & 83.25 & \textbf{84.05} & 89.33 & \textbf{69.43} & 4.987 & \textbf{1.668e-60} & 5.234e-06 \\
4 & \textbf{0.9936} & 0.9807 & \textbf{0.9938} & 0.9813 & \textbf{694.9} & 1670 & \textbf{695.3} & 1670 & \textbf{701.4} & 1676 & \textbf{154.9} & 82.49 & \textbf{5.637e-82} & 4.283e-65 \\
5 & \textbf{0.987} & 0.9664 & \textbf{0.9874} & 0.9674 & \textbf{1328} & 3059 & \textbf{1329} & 3059 & \textbf{1335} & 3065 & \textbf{182.5} & 89.55 & \textbf{2.241e-86} & 2.742e-67 \\
6 & \textbf{0.982} & 0.9552 & \textbf{0.9829} & 0.9573 & \textbf{1216} & 2820 & \textbf{1217} & 2821 & \textbf{1221} & 2826 & \textbf{176.1} & 89.67 & \textbf{3.437e-58} & 8.643e-47 \\
7 & \textbf{0.9801} & 0.9446 & \textbf{0.9814} & 0.9481 & \textbf{968} & 2518 & \textbf{968.9} & 2518 & \textbf{972.3} & 2522 & \textbf{172.5} & 93.08 & \textbf{3.156e-45} & 1.799e-37 \\
8 & \textbf{0.9797} & 0.9393 & \textbf{0.9812} & 0.9439 & \textbf{837.4} & 2345 & \textbf{838.5} & 2347 & \textbf{841.2} & 2349 & \textbf{165.4} & 94.7 & \textbf{3.318e-38} & 2.109e-32 \\
9 & \textbf{0.9848} & 0.9464 & \textbf{0.9861} & 0.9509 & \textbf{723} & 2378 & \textbf{724.2} & 2379 & \textbf{726.5} & 2381 & \textbf{177.2} & 106.1 & \textbf{3.319e-36} & 2.629e-31 \\
10 & \textbf{0.9854} & 0.946 & \textbf{0.9867} & 0.951 & \textbf{629.1} & 2168 & \textbf{630.5} & 2170 & \textbf{632.4} & 2172 & \textbf{172.3} & 106.4 & \textbf{3.38e-33} & 5.096e-29 \\
11 & \textbf{0.9873} & 0.9497 & \textbf{0.9887} & 0.955 & \textbf{431.9} & 1575 & \textbf{433.5} & 1576 & \textbf{434.7} & 1578 & \textbf{156.3} & 101.7 & \textbf{2.761e-28} & 4.096e-25 \\
12 & \textbf{0.9884} & 0.9514 & \textbf{0.9896} & 0.9565 & \textbf{417.8} & 1594 & \textbf{419.4} & 1595 & \textbf{420.7} & 1597 & \textbf{159.6} & 104.4 & \textbf{1.933e-28} & 2.598e-25 \\
13 & \textbf{0.9888} & 0.9518 & \textbf{0.99} & 0.9569 & \textbf{402.3} & 1576 & \textbf{403.9} & 1578 & \textbf{405.1} & 1579 & \textbf{160.1} & 105.3 & \textbf{1.821e-28} & 2.244e-25 \\
14 & \textbf{0.9891} & 0.9522 & \textbf{0.9902} & 0.9573 & \textbf{396.9} & 1580 & \textbf{398.5} & 1582 & \textbf{399.7} & 1583 & \textbf{161.3} & 106.3 & \textbf{1.605e-28} & 1.904e-25 \\
15 & \textbf{0.9892} & 0.9526 & \textbf{0.9903} & 0.9576 & \textbf{397.7} & 1585 & \textbf{399.3} & 1587 & \textbf{400.5} & 1588 & \textbf{161.8} & 106.6 & \textbf{1.523e-28} & 1.823e-25 \\
16 & \textbf{0.989} & 0.952 & \textbf{0.9901} & 0.9571 & \textbf{405.8} & 1610 & \textbf{407.4} & 1611 & \textbf{408.6} & 1613 & \textbf{162} & 106.5 & \textbf{1.504e-28} & 1.867e-25 \\
17 & \textbf{0.9895} & 0.9536 & \textbf{0.9906} & 0.9585 & \textbf{400.9} & 1618 & \textbf{402.5} & 1619 & \textbf{403.7} & 1621 & \textbf{164.2} & 108.1 & \textbf{1.196e-28} & 1.451e-25 \\
18 & \textbf{0.9883} & 0.9507 & \textbf{0.9895} & 0.9559 & \textbf{435.6} & 1685 & \textbf{437.2} & 1687 & \textbf{438.5} & 1688 & \textbf{163.2} & 106.4 & \textbf{1.32e-28} & 1.894e-25 \\
19 & \textbf{0.9886} & 0.9511 & \textbf{0.9898} & 0.9563 & \textbf{412.5} & 1616 & \textbf{414.1} & 1617 & \textbf{415.3} & 1619 & \textbf{161.1} & 105.7 & \textbf{1.653e-28} & 2.106e-25 \\
20 & \textbf{0.9881} & 0.9503 & \textbf{0.9894} & 0.9556 & \textbf{441.5} & 1697 & \textbf{443.1} & 1699 & \textbf{444.3} & 1700 & \textbf{163.2} & 106.2 & \textbf{1.318e-28} & 1.938e-25 \\
\hline
\end{tabular}
    \caption{As in \cref{tab:boston-stat}, for Boston (restricted).}
    \label{tab:boston-stat-rst}
\end{table*}

\begin{table*}[t]
\centering
\begin{tabular}{l|cc|cc|cc|cc|cc|cc|cc}
\hline
 $n_{\rm end}$& \multicolumn{2}{c|}{Adj $R^2$ ($\Uparrow$)} & \multicolumn{2}{c|}{$R^2$ ($\Uparrow$)} & \multicolumn{2}{c|}{AIC ($\Downarrow$)} & \multicolumn{2}{c|}{AICc ($\Downarrow$)} & \multicolumn{2}{c|}{BIC ($\Downarrow$)} & \multicolumn{2}{c|}{$t$-statistic ($\Uparrow$)} & \multicolumn{2}{c}{$p$-value ($\Downarrow$)} \\
 \hline
 & Polylog & Poly & Polylog & Poly & Polylog & Poly & Polylog & Poly & Polylog & Poly & Polylog & Poly & Polylog & Poly \\
\hline
7 & 0.9976 & \textbf{0.9988} & 0.9984 & \textbf{0.9992} & 8.961 & \textbf{6.037} & 20.96 & \textbf{18.04} & 8.337 & \textbf{5.412} & \textbf{29.43} & 11.37 & \textbf{7.935e-06} & 0.0003413 \\
8 & 0.9978 & \textbf{0.999} & 0.9984 & \textbf{0.9993} & 11.29 & \textbf{7.403} & 19.29 & \textbf{15.4} & 11.13 & \textbf{7.241} & \textbf{36.23} & 14.78 & \textbf{3.014e-07} & 2.566e-05 \\
9 & 0.9978 & \textbf{0.9991} & 0.9983 & \textbf{0.9994} & 15.1 & \textbf{9.095} & 21.1 & \textbf{15.1} & 15.34 & \textbf{9.334} & \textbf{43.68} & 18.88 & \textbf{9.636e-09} & 1.426e-06 \\
10 & 0.9974 & \textbf{0.9993} & 0.998 & \textbf{0.9994} & 21.51 & \textbf{10.98} & 26.31 & \textbf{15.78} & 22.1 & \textbf{11.58} & \textbf{50.74} & 23.27 & \textbf{3.023e-10} & 6.875e-08 \\
11 & 0.9967 & \textbf{0.9994} & 0.9973 & \textbf{0.9995} & 32.08 & \textbf{13.04} & 36.08 & \textbf{17.04} & 32.98 & \textbf{13.95} & \textbf{58.97} & 28.85 & \textbf{7.591e-12} & 2.252e-09 \\
12 & 0.9959 & \textbf{0.9995} & 0.9966 & \textbf{0.9996} & 48.69 & \textbf{15.28} & 52.12 & \textbf{18.71} & 49.89 & \textbf{16.47} & \textbf{69.71} & 36.48 & \textbf{1.299e-13} & 4.325e-11 \\
13 & 0.9954 & \textbf{0.9995} & 0.9961 & \textbf{0.9996} & 63.19 & \textbf{18.29} & 66.19 & \textbf{21.29} & 64.64 & \textbf{19.75} & \textbf{75.88} & 41.26 & \textbf{3.86e-15} & 1.675e-12 \\
14 & 0.9947 & \textbf{0.9994} & 0.9955 & \textbf{0.9995} & 82.51 & \textbf{22.51} & 85.17 & \textbf{25.18} & 84.2 & \textbf{24.21} & \textbf{82.86} & 47.08 & \textbf{9.852e-17} & 4.867e-14 \\
15 & 0.994 & \textbf{0.9993} & 0.9948 & \textbf{0.9994} & 106.1 & \textbf{28.41} & 108.5 & \textbf{30.81} & 108.1 & \textbf{30.33} & \textbf{90.65} & 53.81 & \textbf{2.169e-18} & 1.118e-15 \\
16 & 0.9933 & \textbf{0.9991} & 0.9942 & \textbf{0.9992} & 140.9 & \textbf{37.75} & 143.1 & \textbf{39.93} & 143 & \textbf{39.88} & \textbf{102.9} & 64.22 & \textbf{2.577e-20} & 1.173e-17 \\
\hline
\end{tabular}
    \caption{As in \cref{tab:boston-stat}, for Miami (unrestricted endpoint fits).}
    \label{tab:miami-stat}
\end{table*}

\begin{table*}[t]
\centering
\begin{tabular}{l|cc|cc|cc|cc|cc|cc|cc}
\hline
 $w$& \multicolumn{2}{c|}{Adj $R^2$ ($\Uparrow$)} & \multicolumn{2}{c|}{$R^2$ ($\Uparrow$)} & \multicolumn{2}{c|}{AIC ($\Downarrow$)} & \multicolumn{2}{c|}{AICc ($\Downarrow$)} & \multicolumn{2}{c|}{BIC ($\Downarrow$)} & \multicolumn{2}{c|}{$t$-statistic ($\Uparrow$)} & \multicolumn{2}{c}{$p$-value ($\Downarrow$)} \\
 \hline
 & Polylog & Poly & Polylog & Poly & Polylog & Poly & Polylog & Poly & Polylog & Poly & Polylog & Poly & Polylog & Poly \\
\hline
1 & \textbf{0.9999} & 0.9978 & \textbf{0.9999} & 0.9979 & \textbf{-16.21} & 34.4 & \textbf{-15.78} & 34.83 & \textbf{-9.98} & 40.63 & 30.93 & \textbf{157.2} & 2.584e-37 & \textbf{6.889e-77} \\
2 & \textbf{0.9999} & 0.9998 & \textbf{0.9999} & 0.9998 & \textbf{55.2} & 57.73 & \textbf{55.64} & 58.17 & \textbf{61.43} & 63.97 & \textbf{54.14} & 3.378 & \textbf{1.044e-50} & 0.00132 \\
3 & \textbf{0.9994} & 0.9985 & \textbf{0.9994} & 0.9986 & \textbf{118} & 152 & \textbf{118.4} & 152.5 & \textbf{124.2} & 158.2 & \textbf{83.86} & 11.75 & \textbf{2.111e-61} & 7.369e-17 \\
4 & \textbf{0.9985} & 0.9971 & \textbf{0.9985} & 0.9972 & \textbf{204.5} & 290.7 & \textbf{204.9} & 291.1 & \textbf{210.7} & 296.9 & \textbf{112.8} & 16.46 & \textbf{1.067e-68} & 2.545e-23 \\
5 & \textbf{0.9982} & 0.9945 & \textbf{0.9983} & 0.9947 & \textbf{250.3} & 468.4 & \textbf{250.7} & 468.8 & \textbf{256.5} & 474.6 & \textbf{130.5} & 30.74 & \textbf{2.693e-72} & 3.595e-37 \\
6 & \textbf{0.9974} & 0.9908 & \textbf{0.9975} & 0.9912 & \textbf{254.7} & 607.7 & \textbf{255.2} & 608.3 & \textbf{260.2} & 613.2 & \textbf{143.7} & 43.47 & \textbf{1.904e-60} & 8.683e-38 \\
7 & \textbf{0.9971} & 0.9866 & \textbf{0.9973} & 0.9875 & \textbf{196.4} & 622.5 & \textbf{197.2} & 623.3 & \textbf{200.9} & 627 & \textbf{146.1} & 57.72 & \textbf{1.418e-45} & 4.047e-33 \\
8 & \textbf{0.9982} & 0.9875 & \textbf{0.9984} & 0.9884 & \textbf{125.9} & 497.1 & \textbf{127} & 498.1 & \textbf{129.8} & 501 & \textbf{147.7} & 75.52 & \textbf{2.691e-38} & 4.996e-31 \\
9 & \textbf{0.9973} & 0.988 & \textbf{0.9975} & 0.989 & \textbf{141.5} & 439 & \textbf{142.7} & 440.2 & \textbf{145.1} & 442.6 & \textbf{147} & 83.28 & \textbf{2.013e-34} & 5.327e-29 \\
10 & \textbf{0.9968} & 0.9897 & \textbf{0.9971} & 0.9906 & \textbf{146.4} & 368.2 & \textbf{147.7} & 369.6 & \textbf{149.7} & 371.5 & \textbf{144.8} & 88.1 & \textbf{1.089e-31} & 2.217e-27 \\
11 & \textbf{0.9954} & 0.995 & \textbf{0.9959} & 0.9955 & \textbf{162.5} & 174.1 & \textbf{164.1} & 175.7 & \textbf{165.4} & 176.9 & \textbf{131.8} & 84.19 & \textbf{4.995e-27} & 1.003e-23 \\
12 & 0.9948 & \textbf{0.996} & 0.9953 & \textbf{0.9964} & 193.7 & \textbf{155.9} & 195.3 & \textbf{157.5} & 196.5 & \textbf{158.7} & \textbf{138.1} & 91.27 & \textbf{2.262e-27} & 2.551e-24 \\
13 & 0.9941 & \textbf{0.9965} & 0.9948 & \textbf{0.9969} & 216.6 & \textbf{141.9} & 218.2 & \textbf{143.5} & 219.4 & \textbf{144.7} & \textbf{139.4} & 94.11 & \textbf{1.914e-27} & 1.517e-24 \\
14 & 0.9934 & \textbf{0.9973} & 0.9941 & \textbf{0.9976} & 245.1 & \textbf{119.3} & 246.7 & \textbf{120.9} & 247.9 & \textbf{122.2} & \textbf{140.7} & 96.76 & \textbf{1.633e-27} & 9.466e-25 \\
15 & 0.9931 & \textbf{0.9975} & 0.9939 & \textbf{0.9978} & 258.5 & \textbf{114.1} & 260.1 & \textbf{115.7} & 261.3 & \textbf{116.9} & \textbf{142.4} & 98.72 & \textbf{1.331e-27} & 6.729e-25 \\
16 & 0.993 & \textbf{0.9975} & 0.9937 & \textbf{0.9978} & 255.5 & \textbf{111.9} & 257.1 & \textbf{113.5} & 258.3 & \textbf{114.7} & \textbf{141.3} & 97.7 & \textbf{1.535e-27} & 8.031e-25 \\
\hline
\end{tabular}
    \caption{As in \cref{tab:boston-stat}, for Miami (restricted).}
    \label{tab:miami-stat-rst}
\end{table*}

As a heuristic summary of the five primary model-comparison diagnostics (adjusted $R^2$, $R^2$, AIC, AICc, and BIC), \cref{tab:count-stat} favors the polylog model on Boston for unrestricted endpoint fits across $n_{\rm end}=11,12,\dots,20$, with $n_{\rm end,c}$ defined as discussed in the main text. On Boston restricted-HW fits, it also favors the polylog model across $w=4,5,\dots,20$.
On Miami unrestricted endpoint fits, the same summary favors the poly model across $n_{\rm end}=10,11,\dots,16$. On Miami restricted-HW fits, it favors the polylog model for $w=3,4,\dots,11$, then transitions to the poly model for $w=12,13,\dots,16$.

\begin{table*}[t]
    \centering
\begin{tabular}{l|l|l |l|*{20}{c}}
\hline
\hline
& Type & threshold(s) & model & 1 & 2 & 3 & 4 & 5 & 6 & 7 & 8 & 9 & 10 & 11 & 12 & 13 & 14 & 15 & 16 &17&18&19&20\\
\hline
\multirow{4}{*}{Boston}&\multirow{2}{*}{unrestricted}
&\multirow{2}{*}{$n_{\rm end,c}=11$}& Polylog &  &  &  &  &  &  & 0 & 0 & 0 & 0 & 5 & 5 & 5 & 5 & 5 & 5 &5&5&5&5\\
  &&& Poly    &  &  &  &  &  &  & 5 & 5 & 5 & 5 & 0 & 0 & 0 & 0 & 0 & 0 &0&0&0&0\\
\cline{2-24}
&\multirow{2}{*}{restricted}
&\multirow{2}{*}{$w_c=4$}& Polylog & 5 & 0 & 5 & 5 & 5 & 5 & 5 & 5 & 5 & 5 & 5 & 5 & 5 & 5 & 5 & 5 &5&5&5&5\\
  &&& Poly    & 0 & 5 & 0 & 0 & 0 & 0 & 0 & 0 & 0 & 0 & 0 & 0 & 0 & 0 & 0 & 0 &0&0&0&0\\
\hline
\multirow{4}{*}{Miami}&\multirow{2}{*}{unrestricted}
&\multirow{2}{*}{$n_{\rm end,c}=10$}& Polylog &  &  &  &  &  &  &0 & 0 & 0 & 0 & 0 & 0 & 0 & 0 & 0 & 0 &&&&\\
  &&& Poly    &  &  &  &  &  &  & 5 & 5 & 5 & 5 & 5 & 5 & 5 & 5 & 5 & 5 &&&&\\
\cline{2-24}
&\multirow{2}{*}{restricted}
&\multirow{2}{*}{$w=3,12$}& Polylog & 5 & 5 & 5 & 5 & 5 & 5 & 5 & 5 & 5 & 5 & 5 & 0 & 0 & 0 & 0 & 0 &&&&\\
  &&& Poly    & 0 & 0 & 0 & 0 & 0 & 0 & 0 & 0 & 0 & 0 & 0 & 5 & 5 & 5 & 5 & 5 &&&&\\
\hline
\end{tabular}
    \caption{
    Counts of the five primary model-comparison diagnostics favoring each model. The numerical columns are fit-coordinate values: $n_{\rm end}$ for unrestricted endpoint fits and $w$ for restricted-HW fits. The threshold column reports the onset of the strongly favored regime; for Miami restricted-HW data, both transition points are shown because the preferred model changes from polylog to poly at larger $w$.}
    \label{tab:count-stat}
\end{table*}

\cref{fig:grid-plot} provides the $\NTS_Q$ data along with the polylog and poly model fits for Hamming weights $w=2,\dots,16$ on both Boston and Miami.

\section{Interpretation of AIC differences via Akaike weights}
\label{app:AIC}

AIC differences are most naturally interpreted through Akaike weights, which quantify the relative support for each candidate model in an information-theoretic sense \cite{Wagenmakers:2004aa,Burnham:book}. For a set of candidate models $\{M_i\}_{i=1}^R$ with AIC values $\mathrm{AIC}_i$, define
\begin{equation}
\Delta_i \equiv \mathrm{AIC}_i-\mathrm{AIC}_{\min},
\end{equation}
where $\mathrm{AIC}_{\min}$ is the smallest AIC among the $R$ models. The relative likelihood of model $M_i$ is then
\begin{equation}
\mathcal{L}_i=\exp\!\left(-\frac{\Delta_i}{2}\right),
\end{equation}
and the corresponding Akaike weight is
\begin{equation}
w_i=\frac{\exp(-\Delta_i/2)}{\sum_{r=1}^{R}\exp(-\Delta_r/2)}.
\label{eq:akaike-weight}
\end{equation}
The weights satisfy $0\le w_i\le 1$ and $\sum_i w_i=1$, and may therefore be interpreted as probability-like measures of support for each model within the candidate set.

In the two-model comparison used in this work, we define the signed AIC difference as
\begin{equation}
\Delta \mathrm{AIC}\equiv \mathrm{AIC}(\text{poly})-\mathrm{AIC}(\text{polylog}).
\end{equation}
If $\Delta \mathrm{AIC}>0$, then the polylog model has the lower AIC and is preferred; if $\Delta \mathrm{AIC}<0$, then the poly model is preferred. In either case, the relative likelihood of the worse model is
\begin{equation}
\mathcal{L}_{\mathrm{worse}}=\exp\!\left(-\frac{|\Delta \mathrm{AIC}|}{2}\right).
\label{eq:relative-likelihood-worse}
\end{equation}

For our rule-of-thumb threshold $|\Delta \mathrm{AIC}|=10$, this becomes
\begin{equation}
\mathcal{L}_{\mathrm{worse}}=\exp(-5)=6.74\times 10^{-3}\approx 0.0067.
\end{equation}
Thus, the worse model has less than $1\%$ of the support of the better model. In the two-model case, the corresponding Akaike weights are
\begin{equation}
w_{\mathrm{better}}=\frac{1}{1+e^{-5}}\approx 0.9933,
\qquad
w_{\mathrm{worse}}=\frac{e^{-5}}{1+e^{-5}}\approx 0.0067.
\end{equation}
This provides a justification for the criterion $|\Delta \mathrm{AIC}|\gtrsim 10$~\cite{Wagenmakers:2004aa}: the lower-AIC model is then overwhelmingly favored in the Akaike-weight sense.

\section{NISQ post-processing algorithm}
\label{app:algo-nisq}

Here we present the algorithm (\cref{algo:new}) used by our NISQ player for post-processing the stream of $z$ values: each $z$ is a bitstring representing the measurement outcome from executing the Simon circuit for a given hidden bitstring $b$.
The algorithm is similar to \cite[Appendix I]{PhysRevX.15.021082} but corrects a small inaccuracy in the handling of $z = 0$ bitstrings in the computation of $\hat{f}(i)$ (defined below).
Additionally, its implementation is improved to allow us to extend our analysis to larger values of $n$ than what was possible using the original implementation of Ref.~\cite{PhysRevX.15.021082}.

Let $Z$ be a random variable indicating the next observed bitstring in the stream,
let $B$ be the random variable describing the true value of the hidden bitstring $b \in \mathcal{B}$,
where $\mathcal{B} = \{b \in \{0, 1\}^n: \HW(b) \in [1, w]\}$ is the set of all possible options for the hidden bitstring.
The algorithm is based on keeping track of all posterior probabilities $\Pr(B=b)$ for all $b \in \mathcal{B}$. Initially, the prior probability distribution of $B$ is uniform. According to Bayes' formula, the update on observing the new value $z$ is
\begin{equation}
\Pr(B = b | Z = z) = 2^{n-1} \Pr(Z = z| B=b) \Pr(B=b) / C,
\label{eq:bayes-update-def}
\end{equation}
where $C$ is a normalization factor to ensure $\sum_{b} \Pr(B = b| Z = z) = 1$. 
Here we introduced a constant $2^{n-1}$ for convenience in the computations below; it does not affect the correctness of Bayes' formula, since it can be absorbed into $C$.
Following the exponential error model presented in \cite[Section IV.B.1]{PhysRevX.15.021082}, we 
use a likelihood model in which $\Pr(Z = z | B = b)$ depends only on $z \cdot b$ and $\HW(b)$. In particular, 
conditioned on $B = b$, 
all bitstrings satisfying $z\cdot b=0$ are assigned the same probability, including $z=0$, and all bitstrings satisfying $z\cdot b=1$ are assigned the same probability. Thus, there is a function $f\colon \{1, \dots, w\} \to [0, 1]$ such that
\begin{equation}
  2^{n-1} \Pr(Z = z| B = b) = \begin{cases}
    f(\HW(b)) \textrm{ if } z \cdot b = 0,\\
    1 - f(\HW(b)) \textrm{ if } z \cdot b = 1.
  \end{cases}
  \label{eq:Z-given-B-model}
\end{equation}
The function $f(i)$ describes the quality of the device: $f(i) = 1$ for an ideal quantum device and $f(i) = 0.5$ for a device returning uniformly random measurement results independently of the circuit given to it. The posterior probabilities computed by \cref{algo:new} should therefore be understood as posterior probabilities under this likelihood model.

\begin{algorithm}[H]
\caption{Solving for $b$ 
using estimates $\hat{f}(i)$ 
of
$f(i)=\Pr(z\cdot b=0 \mid \HW(b)=i)$ for $i=1,\dots,w$}

\label{algo:new}
\begin{algorithmic}
\Require $n,w$, stream $\mathcal{Z}=\{z_1,z_2,\dots,z_m\}$ of bitstrings, $\{\hat{f}(i)\}_{i=1,\dots,w}$, threshold $\theta \in (0, 1)$.
\Ensure A maximum-a-posteriori estimate of $b$ together with the posterior distribution over $\mathcal{B}=\{b_1,b_2,\dots,b_{N_w}\}$
\State Initialize $\mathcal{B}=\{b_1,b_2,\dots,b_{N_w}\}$
\State Initialize $\Pr_{\mathrm{pre}}(b)=\Pr_{\mathrm{post}}(b)=1/N_w$ for all $b\in \mathcal{B}$
\For{$z \in \mathcal{Z}$}
    \For{$b \in \mathcal{B}$}
        \If{$z\cdot b = 0$}
            \State $\Pr_{\mathrm{post}}(b)\gets \Pr_{\mathrm{pre}}(b)\,\hat{f}(\HW(b))$
        \Else
            \State $\Pr_{\mathrm{post}}(b)\gets \Pr_{\mathrm{pre}}(b)\,[1-\hat{f}(\HW(b))]$
        \EndIf
    \EndFor
    \State $C\gets \sum_{b\in \mathcal{B}}\Pr_{\mathrm{post}}(b)$
    \For{$b \in \mathcal{B}$}
        \State $\Pr_{\mathrm{post}}(b)\gets \Pr_{\mathrm{post}}(b)/C$
        \State $\Pr_{\mathrm{pre}}(b)\gets \Pr_{\mathrm{post}}(b)$
    \EndFor
    \If{$\max_{b\in \mathcal{B}}\Pr_{\mathrm{post}}(b) \geq \theta$}
        \State \textbf{break}
    \EndIf
\EndFor
\State \Return $\arg\max_{b\in \mathcal{B}}\Pr_{\mathrm{post}}(b)$, $\{\Pr_{\mathrm{post}}(b)\}_{b \in \mathcal{B}}$
\end{algorithmic}
\end{algorithm}

We estimate $f(i)$ using $\hat{f}(i)$ given by
\begin{equation}
  \hat{f}(i) = \frac{\textrm{number of shots with $\HW(b) = i$ and $z \cdot b = 0$}}{\textrm{number of shots with $\HW(b) = i$}}.
  \label{eq:hatfi-def}
\end{equation}
To avoid the risk of data snooping, the estimate is computed from
calibration experiments performed separately from the experiments
that generate the stream $\mathcal{Z} = \{z_1,z_2,\dots,z_m\}$ given to the algorithm. Our posterior update uses $\hat{f}(i)$ in place of (unknown) $f(i)$. Since, under our model, $f(i)$ is expected to lie in $[0.5, 1]$, and $f(i) = 0.5$ provides no information for distinguishing hidden strings of Hamming weight $i$, we do not compute $\NTS_Q(n;w)$ for devices and pairs $(n, w)$ for which $\max_{1\le i\le w} \hat{f}(i) \leq 0.5$ (see \cref{fig:summary-range}).
For each value of $n$ shown, the cross markers in \cref{fig:z=0} indicate the weighted average of $\hat f(i)$ over the Hamming weights $i = 1, \dots, n$, namely $\frac{\sum_{i=1}^{n} \binom{n}{i}\hat f(i)}{2^n-1}$.

\begin{figure}[h]
    \centering
    \includegraphics[width=0.48\textwidth]{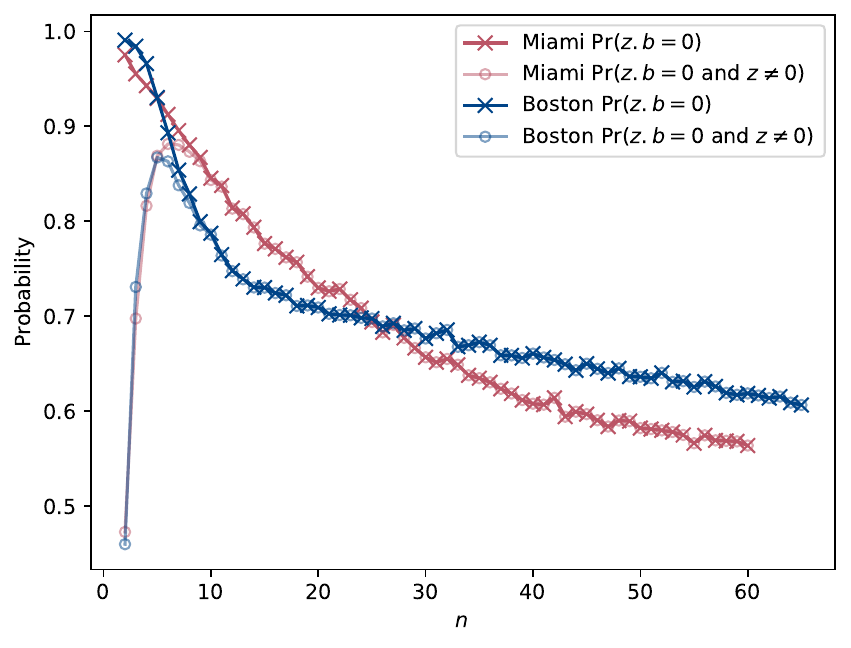}
    
    \caption{Probabilities of obtaining a valid equation, $\Pr(z \cdot b = 0)$ (cross markers) and a valid nonzero equation, $\Pr(z\cdot b = 0\land z\neq 0)$ (open-circle markers) on Miami (red) and Boston (blue), as functions of problem size $n=2,3,\dots,n_{\max}$, where $n_{\max} = 60$ (Miami) and $n_{\max} = 65$ (Boston). Each data point is averaged over Hamming-weight classes $i = 1, 2, \dots, n$ with weights proportional to $\binom{n}{i}$.
    The probability $\Pr(z \cdot b = 0)$ serves as a direct measure of circuit performance, since an ideal Simon circuit has $\Pr(z \cdot b = 0) = 1$.}
    \label{fig:z=0}
\end{figure}

The player starts by assigning equal prior probability to each candidate hidden string $b$. After each circuit output $z_j$, the player updates the posterior probabilities according to
\begin{equation}
    \Pr_{\textrm{post}}(B = b) := g(z \cdot b; \HW(b)) \Pr_{\textrm{pre}}(B = b) / C,
\end{equation}
where $g(z \cdot b; \HW(b)) = \hat{f}(\HW(b))$ if $z \cdot b = 0$ and $1 - \hat{f}(\HW(b))$ otherwise, and $C$ is the normalization constant to ensure the new probabilities sum to $1$.
The prior for the next query (i.e., next $z$) is then set equal to the posterior from the current update.

Once the posterior probability of one of the bitstrings $b$ reaches a threshold $\theta = 0.8$, our NISQ player stops consuming the queries and guesses that $b$.

The exhaustive posterior update in \cref{algo:new} scales with $N_w$ and becomes memory-limited in the unrestricted case. More efficient information-set-decoding-style post-processing could likely extend the accessible range, but developing and benchmarking such post-processing is outside the scope of this work.

\begin{figure}[h]
    \centering
    \includegraphics[width=0.48\textwidth]{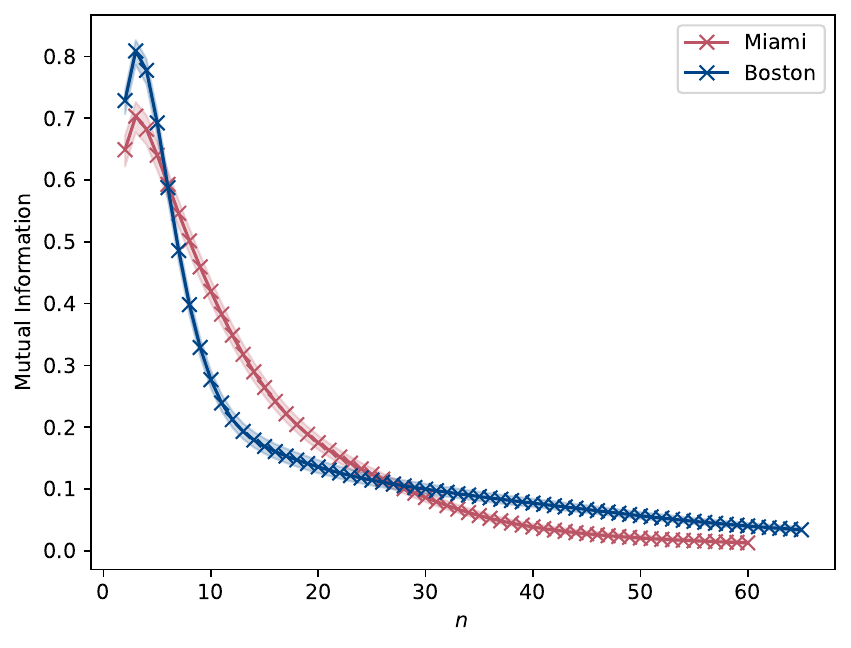}
    \caption{Mutual information obtained per shot on Miami (red) and Boston (blue), as functions of problem size $n=2,3,\dots,n_{\max}$ for the unrestricted Simon problem, where $n_{\max}=60$ on Miami and $n_{\max}=65$ on Boston.}
    \label{fig:mutual_info}
\end{figure}

\begin{figure*}[t]
    \centering
    \includegraphics[width=0.48\textwidth]{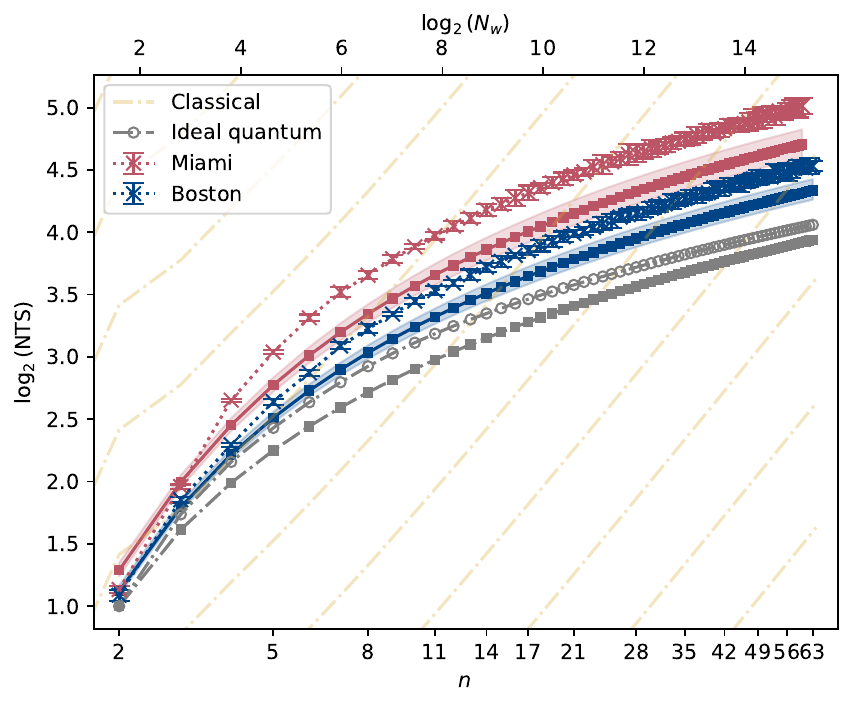}
    \includegraphics[width=0.48\textwidth]{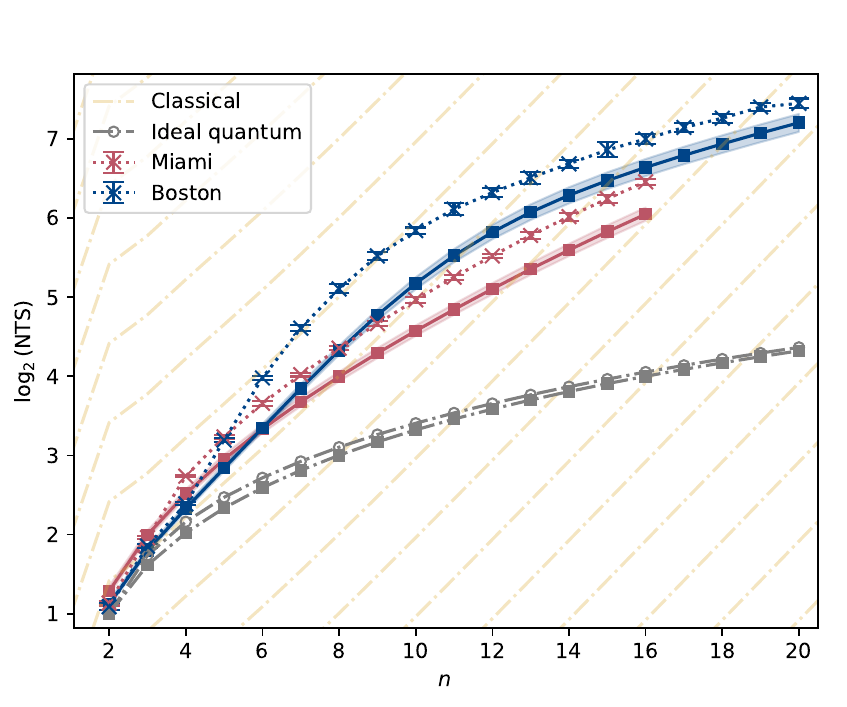}
\caption{Estimates of $\NTS_Q$ via mutual information analysis. Cross markers show direct post-processing estimates for the restricted-HW case with $w=3$, and circle markers show direct post-processing estimates for the unrestricted case $w=n$, both without the polylog and poly model fits. Square markers denote the mutual-information-based estimates for Miami (red), Boston (blue), and the ideal quantum case (gray), consistent with the colors in the legend. Shaded regions indicate $95\%$ Bayesian credible intervals due to the finite number of shots used to estimate $f(i)$.} 
    \label{fig:ntsq-tmi}
\end{figure*}

\section{Results with the help of Q-CTRL}
We also ran experiments on both Boston and Miami using Q-CTRL Fire Opal \cite{fire_opal}, accessed via the Qiskit Function interface. Q-CTRL applies a suite of automated optimizations to IBM QPUs, including circuit depth reduction, error-aware hardware mapping, DD for crosstalk suppression, optimized gate replacement, and MEM \cite{PhysRevApplied.20.024034}. As reported below, the effect of Q-CTRL is device- and metric-dependent. The Q-CTRL data provide a useful supplementary comparison, but because Fire Opal operates as a black box, the details of its internal optimizations are inaccessible to the user. As a result, these results should not be interpreted as part of the compiler-restricted speedup demonstration, as we are unable to verify that the Q-CTRL experiments comply with the compilation restrictions defined in \cref{app:rules:oracle}.

We performed the same experiments on Boston and Miami using Q-CTRL, where the only input required to Fire Opal was the quantum circuit constructed 
using the method described in \cref{sec:o1-oracle,app:simon-oracle}
with all subsequent optimizations handled automatically. We first examine the 
probability $\Pr(z \cdot b = 0)$, which provides a direct measure of 
whether the circuit is functioning correctly. 
The corresponding Q-CTRL results are shown in \cref{fig:z=0-qctrl}.

\begin{figure}[H]
    \centering
    \includegraphics[width=\linewidth]{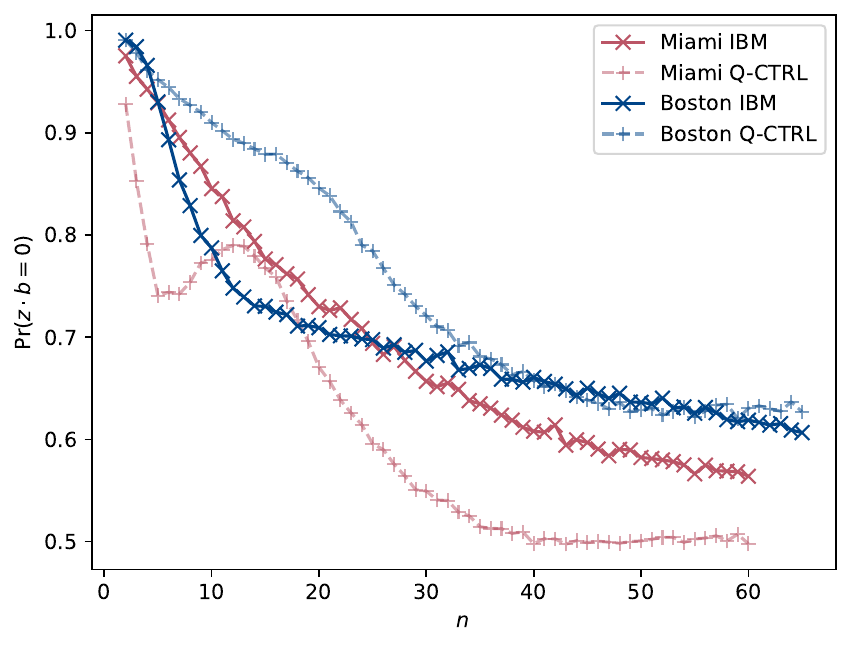}
    \caption{Probability of obtaining $z \cdot b = 0$ as a function of problem size $n$, reproducing the results of \cref{fig:z=0} with the addition of data obtained via Q-CTRL.}
    \label{fig:z=0-qctrl}
\end{figure}

The Q-CTRL results are device-dependent:
on Boston they remain usable over a broad range of problem sizes, whereas on Miami they degrade performance in most cases.
For Miami, the Q-CTRL results yield meaningful estimates only for $w = 1$ and $w = 2$; we therefore restrict the remaining Q-CTRL analysis to Boston.

\Cref{fig:main_nmax_qctrl,fig:main_full_qctrl,fig:R2-AIC-qctrl}
are the Boston analogues of \cref{fig:main-plot-nmax,fig:main-plot-original,fig:R2-AIC} with Q-CTRL data added. The NTS plots show that Q-CTRL changes both the prefactors and the fitted trends. In particular, the absolute NTS values are not uniformly improved relative to the baseline, so these data should be viewed as a supplementary robustness comparison rather than as an improvement over the compiler-restricted results.
\Cref{fig:R2-AIC-qctrl} further shows that, on Boston, the Q-CTRL data are broadly consistent with the conclusion that the polylog model is competitive or favored over much of the plotted range, although the AIC separations are smaller in some regimes. We therefore treat the Q-CTRL results as a robustness comparison rather than as an independent compiler-restricted speedup claim.

\Cref{fig:grid-plot-qctrl} provides the $\NTS_Q$ data, along with the polylog and poly model fits, for Hamming weights $w\in[1,15]$ on Boston with and without Q-CTRL. The maximum problem size used in the fit for each $w$ is shown in \cref{fig:summary-range-qctrl}.

\begin{figure}[H]
    \centering
    \includegraphics[width=\linewidth]{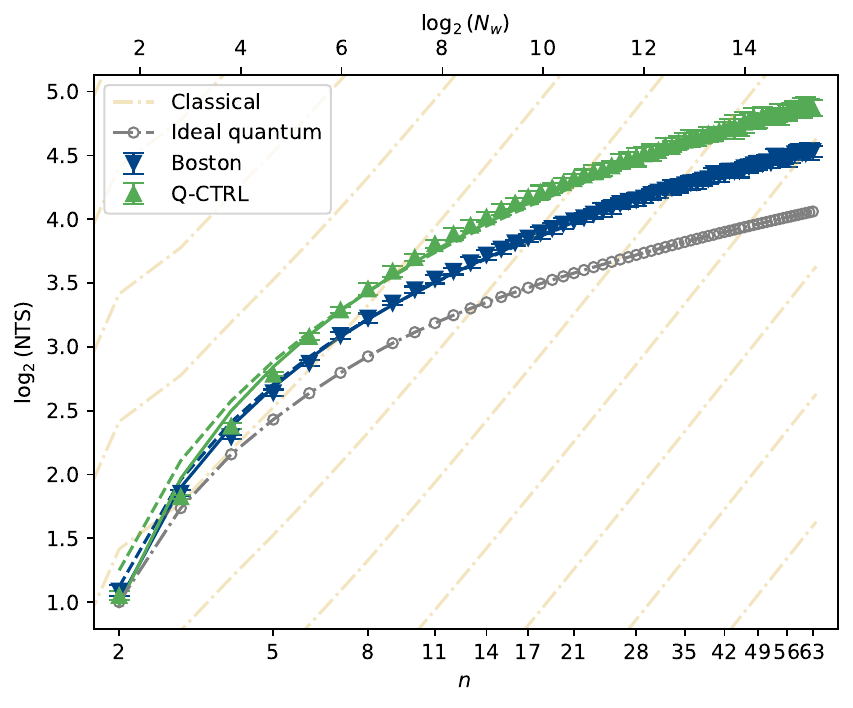}
    \caption{As in \cref{fig:main-plot-nmax} for Boston but with Q-CTRL results added.}
    \label{fig:main_nmax_qctrl}
\end{figure}

\begin{figure}[H]
    \centering
    \includegraphics[width=\linewidth]{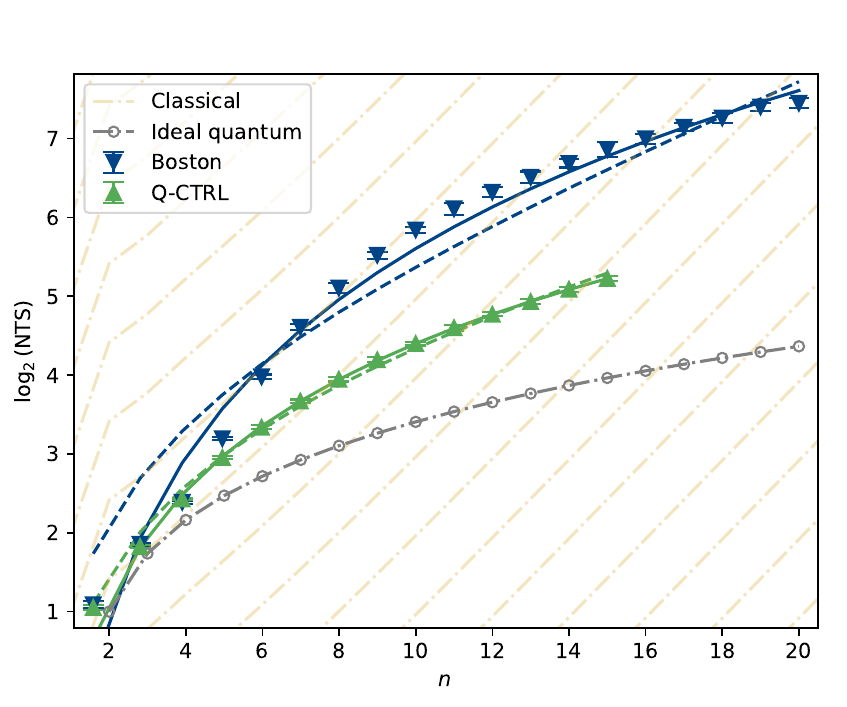}
    \caption{As in \cref{fig:main-plot-original} for Boston but with Q-CTRL results added.}
    \label{fig:main_full_qctrl}
\end{figure}

\begin{figure*}[t]
    \centering
    \includegraphics[width=0.9\textwidth]{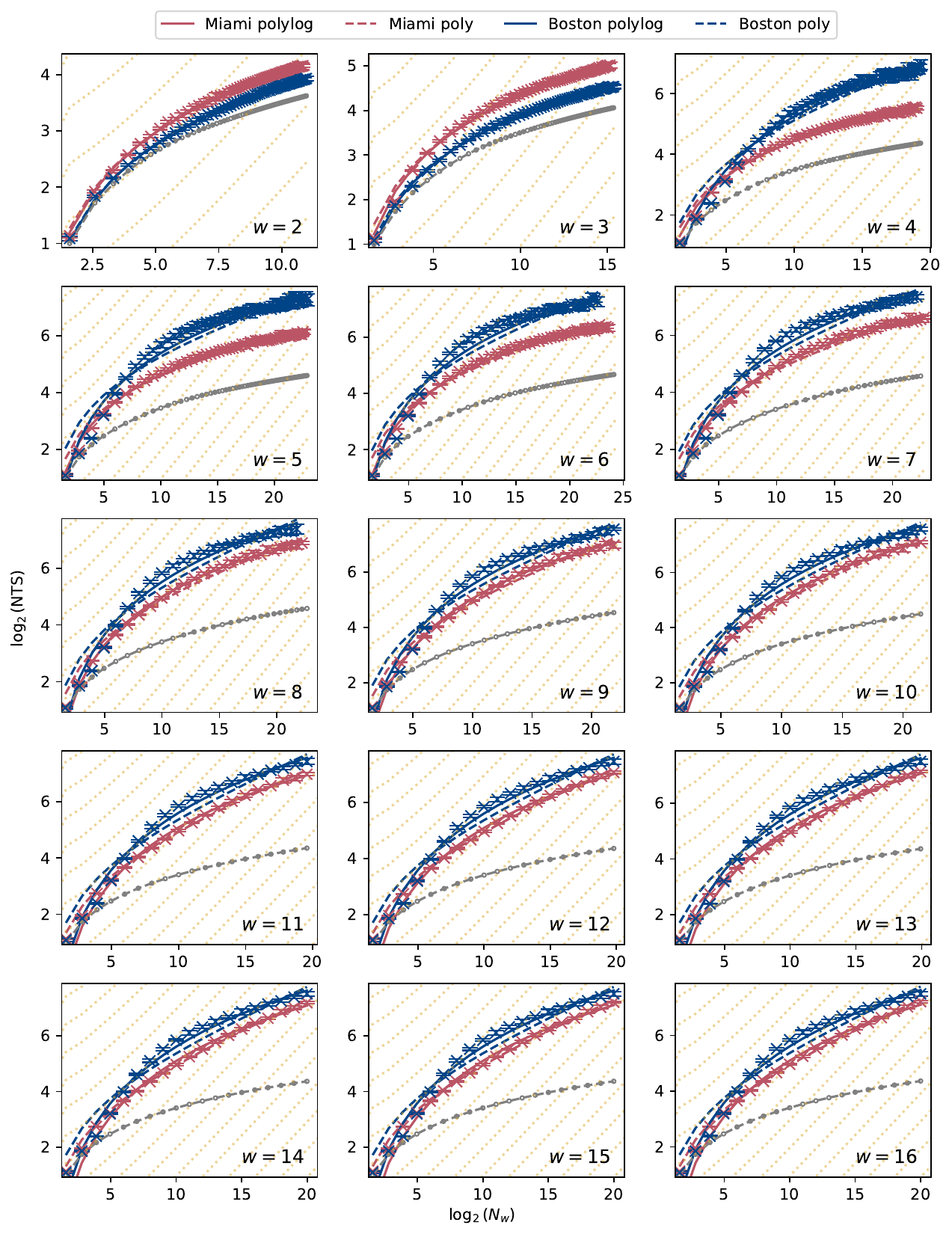}
    \caption{$\log_2(\NTS_Q)$ as a function of $\log_2(N_w)$ for \wSimon{w}{n}
    and $w\in[2,16]$, for Boston (blue) and Miami (red).}
    \label{fig:grid-plot}
\end{figure*}

\begin{figure*}[t]
    \centering
    \includegraphics[width=0.49\linewidth]{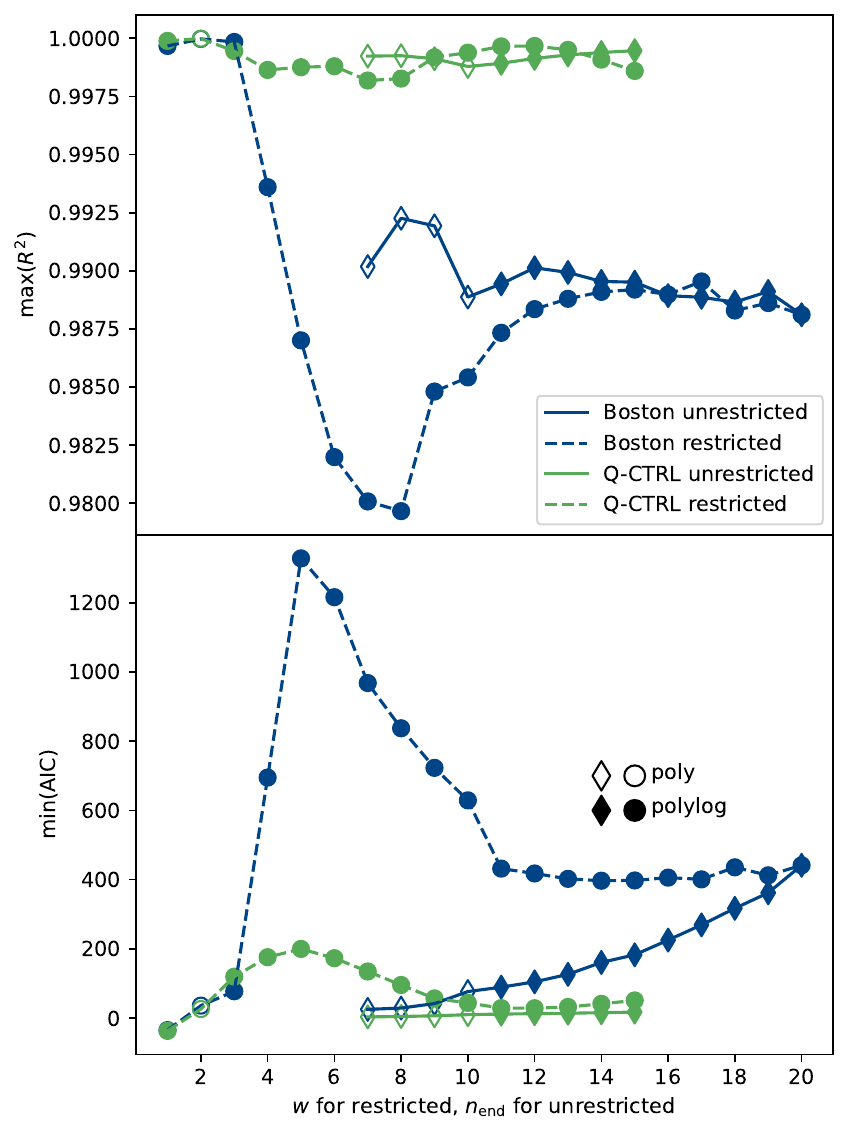}
    \includegraphics[width=0.48\linewidth]{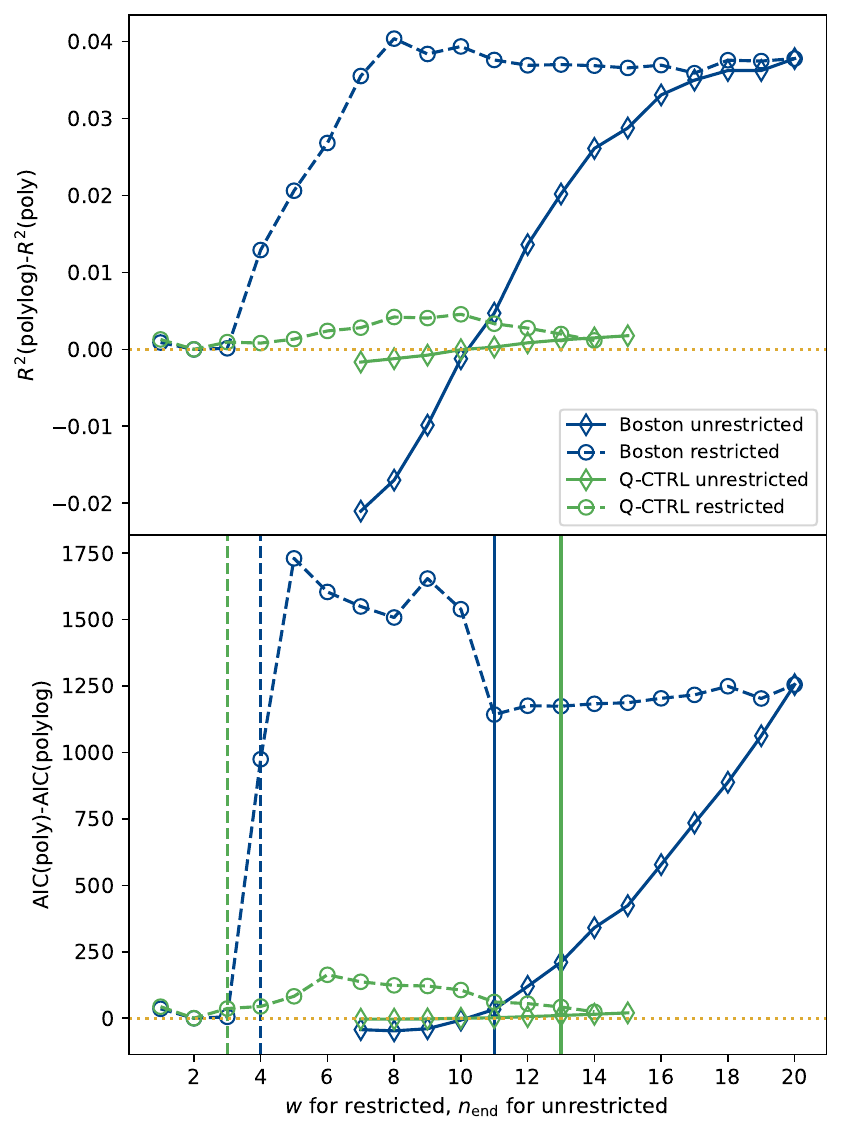}
\caption{Model-comparison diagnostics for the Boston data with and without Q-CTRL, shown in the same format as \cref{fig:R2-AIC}. Blue curves denote the baseline Boston data, and green curves denote the Q-CTRL data. Solid lines with diamond markers correspond to unrestricted Simon endpoint fits, for which the horizontal coordinate is $n_{\rm end}$; dashed lines with circle markers correspond to restricted-HW fits, for which the horizontal coordinate is $w$. When the two finite-window fits have nearly indistinguishable residuals, $\Delta\mathrm{AIC}$ becomes small; this occurs for some of the Q-CTRL fits and should be interpreted as weak model discrimination rather than as decisive evidence for either model.}
    \label{fig:R2-AIC-qctrl}
\end{figure*}

\begin{figure*}
    \centering
    \includegraphics[width=\linewidth]{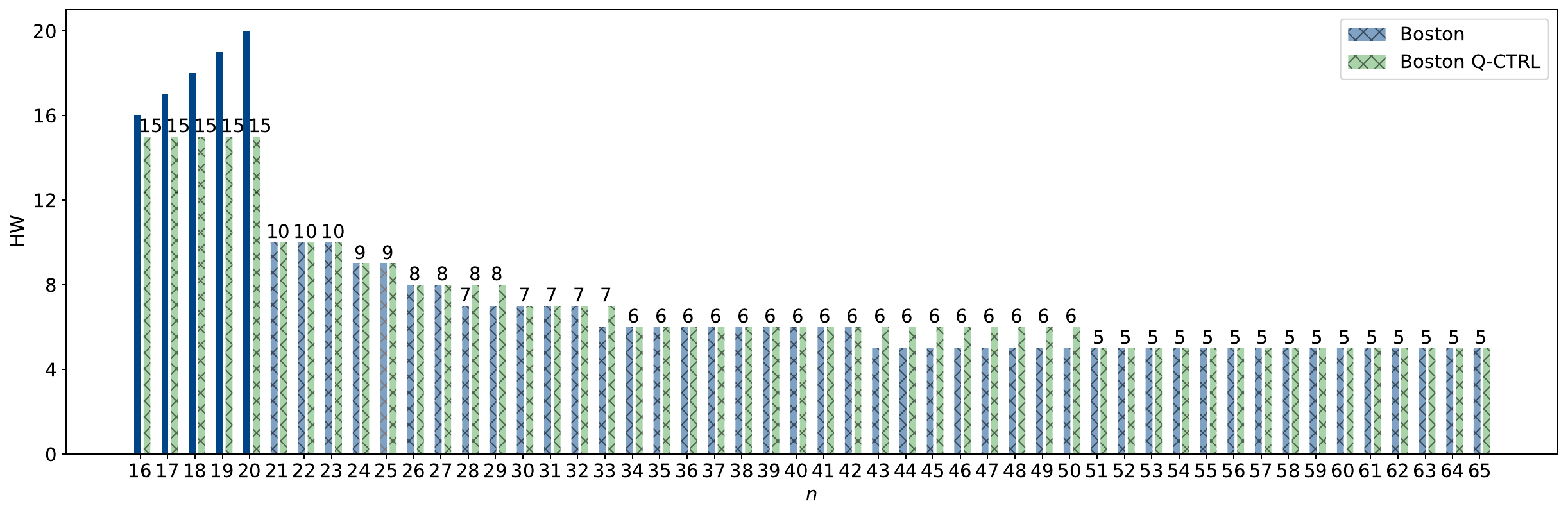}
    \caption{As in \cref{fig:summary-range}, for Boston with and without Q-CTRL.}
    \label{fig:summary-range-qctrl}
\end{figure*}

\begin{figure*}[t]
    \centering
    \includegraphics[width=0.9\textwidth]{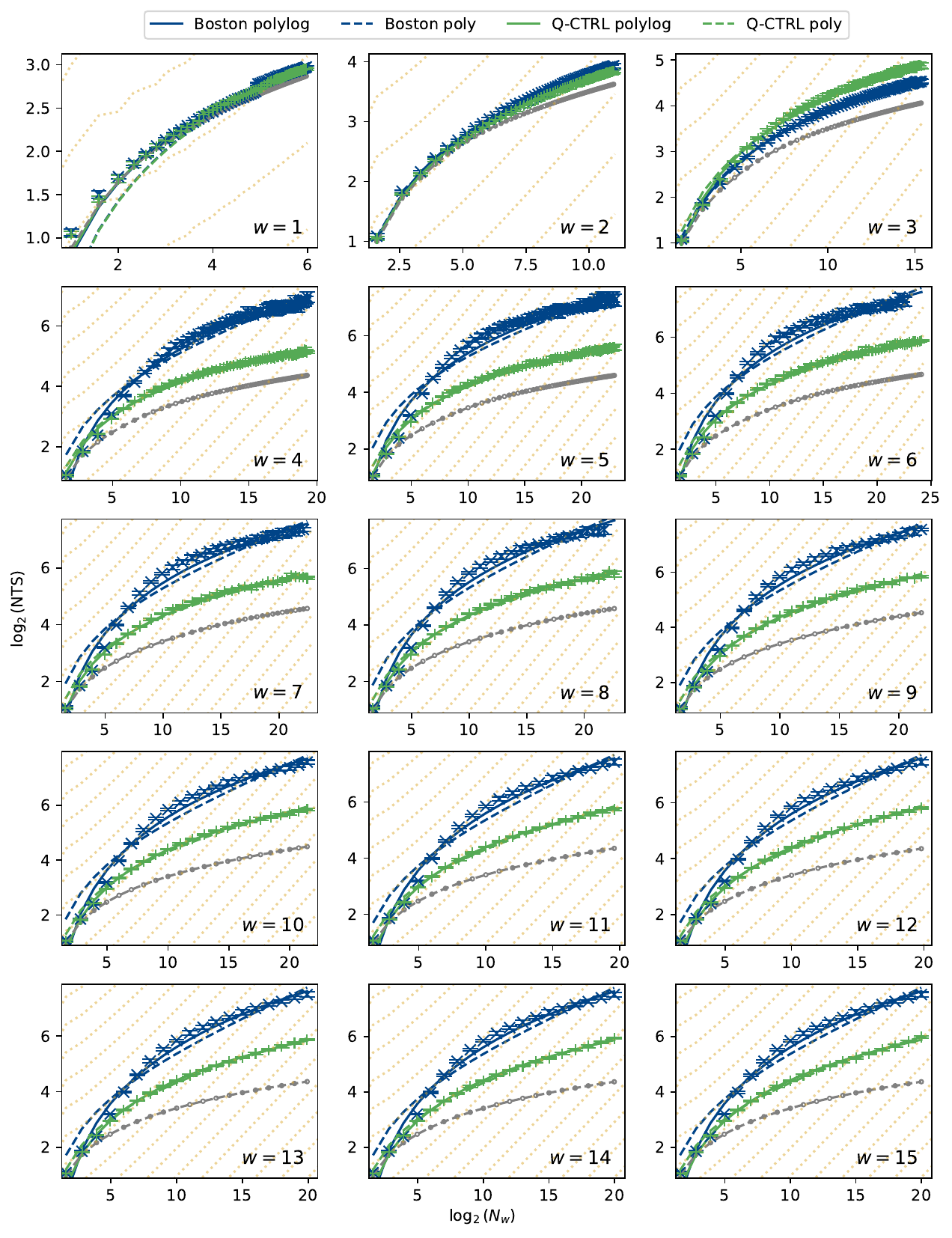}
    \caption{$\log_2(\NTS_Q)$ as a function of $\log_2(N_w)$ for HW cutoffs $w\in[1,15]$, for Boston with Q-CTRL (green, plus markers) and without Q-CTRL (blue, cross markers).}
    \label{fig:grid-plot-qctrl}
\end{figure*}

\section{Mutual information}\label{app:mutual-info}
In this section we compute the classical mutual information $I(Z; B)$
which can be learned by the NISQ player from the first query in \wSimon{w}{n} problem under the same model \cref{eq:Z-given-B-model}
as used by \cref{algo:new}.
Similarly to \cite[Appendix H]{PhysRevX.15.021082},
we can express $I(Z; B)$ as
\begin{equation}
    \begin{aligned}
    I(Z;B)&=S(Z)-S(Z|B) \\
     &=\sum_b \Pr(b)\sum_{z} \Pr(z|b) \log_2 \frac{\Pr(z|b)}{\Pr(z)},
    \end{aligned}
    \label{eq:info}
\end{equation}
where $S(X)$ denotes the Shannon entropy of a random variable $X$.
However, unlike \cite[Appendix H]{PhysRevX.15.021082}, we cannot factor
out the sum over $b$ because the inner sum depends on the Hamming weight of $b$.

Let $s_2: [0,1]\to[0,1]$ be the binary entropy function, defined by $s_2(p)=-p\log_2 p-(1-p)\log_2(1-p)$. Here and below we define $x \log_2(x)$ to be equal to $0$ for $x = 0$.
Then we can compute
\begin{equation}
  S(Z|B) = (n-1) + \frac{1}{N_w} \sum_{i=1}^{w} \binom{n}{i} s_2(f(i)).
  \label{eq:SZB-formula}
\end{equation}
Let $k = \HW(z)$. Then
\begin{multline}
  \Pr(Z=z) = 2^{-n} + 2^{-n} \frac{1}{N_w} \sum_{i=1}^{w} (2f(i) - 1) K_i(k; n) \\
  = 2^{-n} q_k^{(n,w)},
\end{multline}
where
\begin{equation}
   q_k^{(n,w)} = 1 + \sum_{i=1}^{w} (2 f(i) - 1) K_i(k; n) / N_w,
   \label{eq:qk-def}
\end{equation}
and $K_i(\bullet; n)$ denotes the binary Kravchuk polynomial given by
\begin{equation}
    K_i(x; n) = \sum_{j=0}^{i} (-1)^j \binom{x}{j} \binom{n-x}{i-j}.
    \label{eq:Kixn-def}
\end{equation}
Here $j$ plays the role of the number of bits which are set to $1$ in both $z$ and $b$. \Cref{eq:Kixn-def} uses the convention that $\binom{a}{b} = 0$ for $b < 0$ or $b > a$.
Then
\begin{multline}
    S(Z) = -\sum_{z \in \{0, 1\}^n} \Pr(Z = z) \log_2(\Pr(Z = z)) = \\
    n - \sum_{k=0}^{n} 2^{-n} \binom{n}{k} q_k^{(n,w)} \log_2(q_k^{(n,w)}).
    \label{eq:SZ-formula}
\end{multline}

We now describe how we estimate $I(Z; B)$ in practice. First,
we estimate $f(i)$ using an estimate $\hat{f}(i)$ computed
as in \cref{eq:hatfi-def} but now using all available data
instead of restricting to separate calibration experiments.

We compute all Kravchuk polynomial values with $n \leq n_{\mathrm{max}} = 65$ using the relations
\begin{align}
    K_{i}(k; n) &= K_{i}(k-1; n-1) - K_{i-1}(k-1; n-1) \textrm{ if $k > 0$,}\\
    K_{i}(k; n) &= 0 \textrm{ if $i \notin [0, n]$,}\\
    K_{i}(0; n) &= \binom{n}{i}.
\end{align}
We then evaluate \cref{eq:qk-def}, \cref{eq:SZ-formula}, and \cref{eq:SZB-formula}, and compute $\hat{I}(Z;B)=\hat{S}(Z)-\hat{S}(Z|B)$ as in \cref{eq:info}.

For uncertainty quantification, we compute Bayesian credible intervals by replacing $\hat{f}(i)$ in the computation above with posterior samples under Jeffreys prior.

We also compute a heuristic estimate for $\NTS_Q$ using $I(Z; B)$ by dividing the entropy of the distribution of $B$ ($\log_2(N_w)$) by the amount of information learned on the first query ($I(Z; B)$):
\begin{equation}
    \NTS_{Q, \mathrm{TMI}} = \log_2(N_w) / I(Z; B).
\end{equation}
This estimate is based on the fact that in $T$ queries we can
learn at most $T I(Z; B)$ bits due to tensorization of mutual information
(also known as subadditivity of mutual information under conditional independence, see \cite[Lemma 2]{scarlett2019introductory}):
\begin{equation}
I(B; Z_1, \dots, Z_T) \leq \sum_{t=1}^T I(B; Z_t) = T I(Z; B).
\label{eq:TMI}
\end{equation}
Thus, $\NTS_{Q, \mathrm{TMI}}$ is a lower bound on the number of queries to learn $\log_2(N_w)$ bits of information about $B$, i.e., to learn the value of $B$.

However, it is not a rigorous lower bound on $\NTS_Q$,
because the player can benefit from guessing the value of $B$ before
they are confident in what $B$ is: while the definition of $\expv{P}$ used in \cref{eq:NTS.def,eq:NTS-wn} penalizes incorrect guesses, such penalty is not enough to ensure $\NTS_{Q, \mathrm{TMI}}$ is a lower bound for $\NTS_Q$ as shown in the following example. Let $n = 2$, $w = 1$, i.e. consider a \wSimon{1}{2} problem. Let $f(1) = p \in (0.5, 1)$. Then
\begin{equation}
    \NTS_{Q, \mathrm{TMI}} = 2 / (1 - s_2(p)).
\end{equation}
Since $z = 00$ and $z = 11$ provide no information about $B$,
the player can use the following strategy. Query until $z \notin \{00,11\}$ is obtained, and then guess the most likely $b$ ($b = 01$ for $z = 10$ and $b = 10$ for $z = 01$). For that strategy we have
\begin{equation}
    \NTS_Q = 2 / (2p - 1)
\end{equation}
For $p \in (0.5, 1)$, $2p-1>1-s_2(p)$. Indeed, the function $1-s_2(p)-(2p-1)$ has second derivative $1/[\ln(2)p(1-p)]>0$ and vanishes at $p=0.5$ and $p=1$, so it is negative between these endpoints. Hence, in this example $\NTS_Q < \NTS_{Q,\mathrm{TMI}}$.

Typically, however, we would expect that $\NTS_{Q, \mathrm{TMI}} < \NTS_Q$ because on subsequent queries there is less information about $B$ left to learn, and, hence, less information is learned per query (this is more rigorously stated as \cref{eq:TMI} above).

$I(Z; B)$ and $\NTS_{Q, \mathrm{TMI}}$ are illustrated on \cref{fig:mutual_info} and \cref{fig:ntsq-tmi}, respectively.

\end{document}